\documentclass[11pt]{article}

\usepackage{amsmath,amsfonts,amssymb,bbm,enumitem,natbib,url,color,graphicx,amsthm}
\usepackage[a4paper]{geometry}

\newtheorem{theorem}{Theorem}
\newtheorem{corollary}[theorem]{Corollary}

\newtheorem{lemma}[theorem]{Lemma}

\newcommand{\dd}{\, {\rm d}}

\newcommand{\real}{\mathbb{R}}

\newcommand{\myQ}{\mathbb{Q}}

\newcommand{\cC}{\mathcal{C}}

\newcommand{\cN}{\mathcal{N}}

\newcommand{\AUC}{\mbox{\rm AUC}}
\newcommand{\CEP}{\mbox{\rm CEP}}
\newcommand{\LR}{\mbox{\rm LR}}
\newcommand{\f}{\mbox{\rm FAR}}
\newcommand{\h}{\mbox{\rm HR}}

\newcommand{\betat}{\text{beta }}
\newcommand{\tR}{\textsf{R }}
\newcommand{\hsp}{\hspace{0.3mm}}
\newcommand{\hspt}{\hspace{4.0mm}}
\newcommand{\fNI}{f_{\hsp \rm NI}}
\newcommand{\FNI}{F_{\hsp \rm NI}}

\begin{document}

\begin{center}

{\bf \Large Receiver Operating Characteristic (ROC) Curves}

\bigskip

{\bf Tilmann Gneiting} 

Heidelberg Institute for Theoretical Studies, Heidelberg, Germany \\
Karlsruhe Institute of Technology

\url{tilmann.gneiting@h-its.org} 

\medskip

{\bf Peter Vogel}

Karlsruhe Institute of Technology, Karlsruhe, Germany \\
Heidelberg Institute for Theoretical Studies

\url{peter.vogel3@kit.edu} 

\bigskip

\today

\smallskip

\end{center}

\begin{abstract}

Receiver operating characteristic (ROC) curves are used ubiquitously
to evaluate covariates, markers, or features as potential predictors in
binary problems.  We distinguish raw ROC diagnostics and ROC curves,
elucidate the special role of concavity in interpreting and modelling
ROC curves, and establish an equivalence between ROC curves and
cumulative distribution functions (CDFs).  These results support a
subtle shift of paradigms in the statistical modelling of ROC curves,
which we view as curve fitting.  We introduce the flexible
two-parameter \betat family for fitting CDFs to empirical ROC curves,
derive the large sample distribution of the minimum distance estimator
and provide software in \tR for estimation and testing, including both
asymptotic and Monte Carlo based inference.  In a range of empirical
examples the \betat family and its three- and four-parameter
ramifications that allow for straight edges fit better than the
classical binormal model, particularly under the vital constraint of
the fitted curve being concave.

\end{abstract}

\section{Introduction}  \label{sec:introduction} 

Through all realms of science and society, the assessment of the
predictive ability of real-valued markers or features for binary
outcomes is of critical importance.  To give but a few examples,
biomarkers are used to diagnose the presence of cancer or other
diseases, numerical weather prediction (NWP) systems serve to
anticipate extreme precipitation events, judges need to assess
recidivism in convicts, in information retrieval documents, such as
websites, are to be classified as signal or noise, banks use
customers' particulars to assess credit risk, financial transactions
are to be classified as fraud or no fraud, and email messages are to
be identified as spam or legitimate.  In these and myriads of similar
settings, receiver operating characteristic (ROC) curves are key tools
in the evaluation of the predictive ability of covariates, markers, or
features (Egan et al.~1961, Swets 1973, 1988, Zweig and Campbell 1993,
Fawcett 2006).  Figure~\ref{fig:WoS} documents the astonishing rise in
the use of ROC curves in the scientific literature.  In 2017, nearly
8,000 papers were published that use ROC curves, up from less than 50
per year through 1990 and less than 1,000 papers annually through
2002.

\begin{figure}[t]

\centering

\includegraphics[width = 0.90 \textwidth]{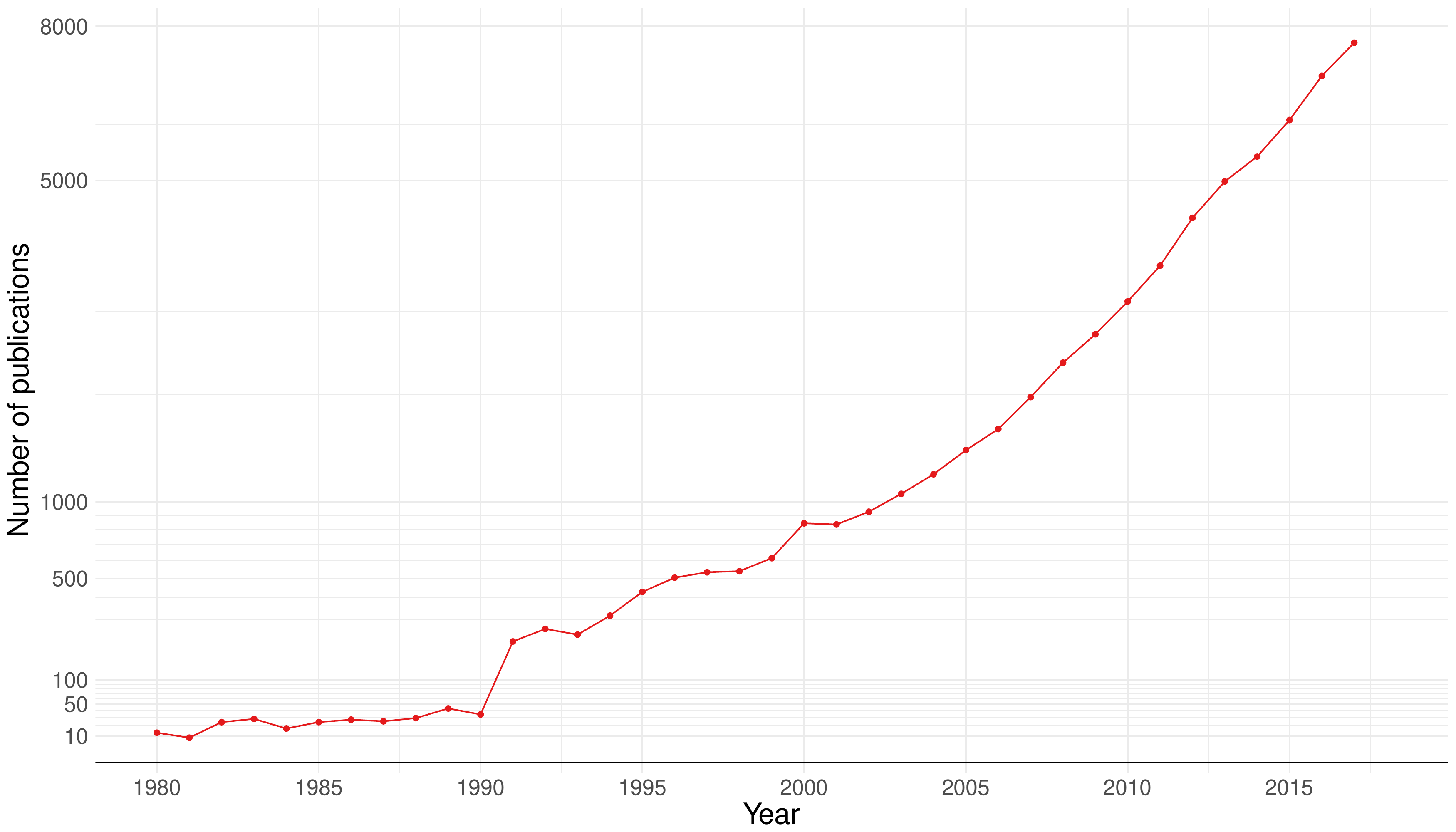}
\caption{Number of publications per year resulting from a Web of
  Science topic search for the terms ``receiver operating
  characteristic'' or ``ROC'' on 24 August 2018.  Note the square root
  scale on the vertical axis, which suggests quadratic
  growth. \label{fig:WoS}}

\end{figure}

A ROC curve is simply a plot of the hit rate (HR) against the false
alarm rate (FAR) across the range of thresholds for the real-valued
marker or feature at hand.  Specifically, consider the joint
distribution $\myQ$ of the pair $(X, Y)$, where the covariate, marker
or feature $X$ is real-valued, and the event $Y$ is binary, with the
implicit understanding that higher values of $X$ provide stronger
support for the event to materialize ($Y = 1$).  The joint
distribution $\myQ$ of $(X,Y)$ is characterized by the {\em
  prevalence}\/ $\pi_1 = \myQ(Y = 1) \in (0,1)$ along with the
conditional cumulative distribution functions (CDFs)
\[
F_1(x) = \myQ( X \leq x \, | \, Y = 1) 
\qquad \textrm{ and } \qquad 
F_0(x) = \myQ( X \leq x \, | \, Y = 0).
\]
Any threshold value $x$ can be used to predict a positive outcome ($Y
= 1$) if $X > x$ and a negative outcome ($Y = 0$) if $X \leq x$, to
yield a classifier with {\em hit rate}\/ (\h),\footnote{Terminologies
  abound and differ markedly between communities. Some researchers
  talk of ROC as {\em relative operating characteristic}; see, e.g.,
  Swets (1973) and Mason and Graham (2002).  The hit rate has also
  been referred to as {\em probability of detection}\/ (POD), {\em
    recall}, {\em sensitivity}, or {\em true positive rate}\/ (TPR).
  The false alarm rate is also known as {\em probability of false
    detection}\/ (POFD), {\em fall-out}, or {\em false positive
    rate}\/ (FPR) and equals one minus the {\em specificity}, {\em
    selectivity}, or {\em true negative rate}\/ (TNR). For an
  overview, see
  \url{https://en.wikipedia.org/wiki/Precision_and_recall#Definition_(classification_context)},
  accessed 21 August 2018.}
\[
\h(x) = \myQ(X > x \, | \, Y = 1) = 1 - F_1(x),
\]
and {\em false alarm rate}\/ (\f), 
\[
\f(x) = \myQ(X > x \, | \, Y = 0) = 1 - F_0(x).  
\]

The term {\em raw ROC diagnostic}\/ refers to the set-theoretic union
of the points of the form $(\f(x), \h(x))'$ within the unit square.
The {\em ROC curve}\/ is a linearly interpolated raw ROC diagnostic,
and therefore it also is a point set that may or may not admit a
direct interpretation as a function.  However, if $F_1$ and $F_0$ are
continuous and strictly increasing, the raw ROC diagnostic and the ROC
curve can be identified with a function $R$, where $R(0) = 0$,
\begin{equation} \label{eq:R(p)} 
R(p) = 1 - F_1(F_0^{-1}(1-p))  
\quad \textrm{ for } \quad 
p \in (0,1), 
\end{equation}
and $R(1) = 1$.  High hit rates and low false alarm rates are
desirable, so the closer the ROC curve gets to the upper left corner
of the unit square the better.  The area under the ROC curve (\AUC) is
a widely used measure of the potential predictive value of a feature
(Hanley and McNeil 1982, 1983, DeLong et al.~1988, Bradley 1997),
admitting an appealing interpretation as the probability of a marker
value drawn from $F_1$ being higher than a value drawn independently
from $F_0$.  Table~\ref{tab:AUC} proposes terminology for the
description of the strength of the potential value in terms of \AUC.

\begin{table}[t]

\caption{Proposed terminology for the potential predictive strength of
  a feature based on the \AUC \ value. \label{tab:AUC}}

\medskip
\centering
\small

\begin{tabular}{cl}
\hline
\AUC & Descriptor \\
\hline
$> 0.99$       & nearly perfect \\
$0.95 - 0.99$ & very strong \\
$0.85 - 0.95$ & strong \\
$0.75 - 0.85$ & substantial \\
$0.65 - 0.75$ & moderate \\
$0.50 - 0.65$ & weak \\
$\leq 0.50$    & abysmal \\
\hline
\end{tabular}

\end{table}

In data analytic practice, the measure $\myQ$ is the empirical
distribution of a sample $(x_i, y_i)_{i=1}^n$ of real-valued features
$x_i$ and corresponding binary observations $y_i$.  To generate a ROC
curve in this setting, it suffices to consider the unique values of
$x_1, \ldots, x_n$ and the respective false alarm and hit rates.  The
resulting raw ROC diagnostic is interpolated linearly to yield an
empirical ROC curve, as illustrated in Figure~\ref{fig:data} on
examples from the biomedical (Etzioni et al.~1999, Sing et al.~2005,
Robin et al.~2011) and meteorological (Vogel et al.~2018) literatures.
Based on \AUC \ and the terminology in Table~\ref{tab:AUC}, the
predictor strength is moderate in the example from Robin et
al.~(2011), substantial for the data from Etzioni et al.~(1999) and
Vogel et al.~(2018), and strong in the example from Sing et
al.~(2005).  Arguably, the immense popularity of empirical ROC curves
and \AUC \ across the scientific literature stems from their ease of
implementation and interpretation in concert with a wide range of
desirable properties, such as invariance under strictly increasing
transformations of a feature.

\begin{figure}[t]

\centering

\includegraphics[width = 0.925 \textwidth]{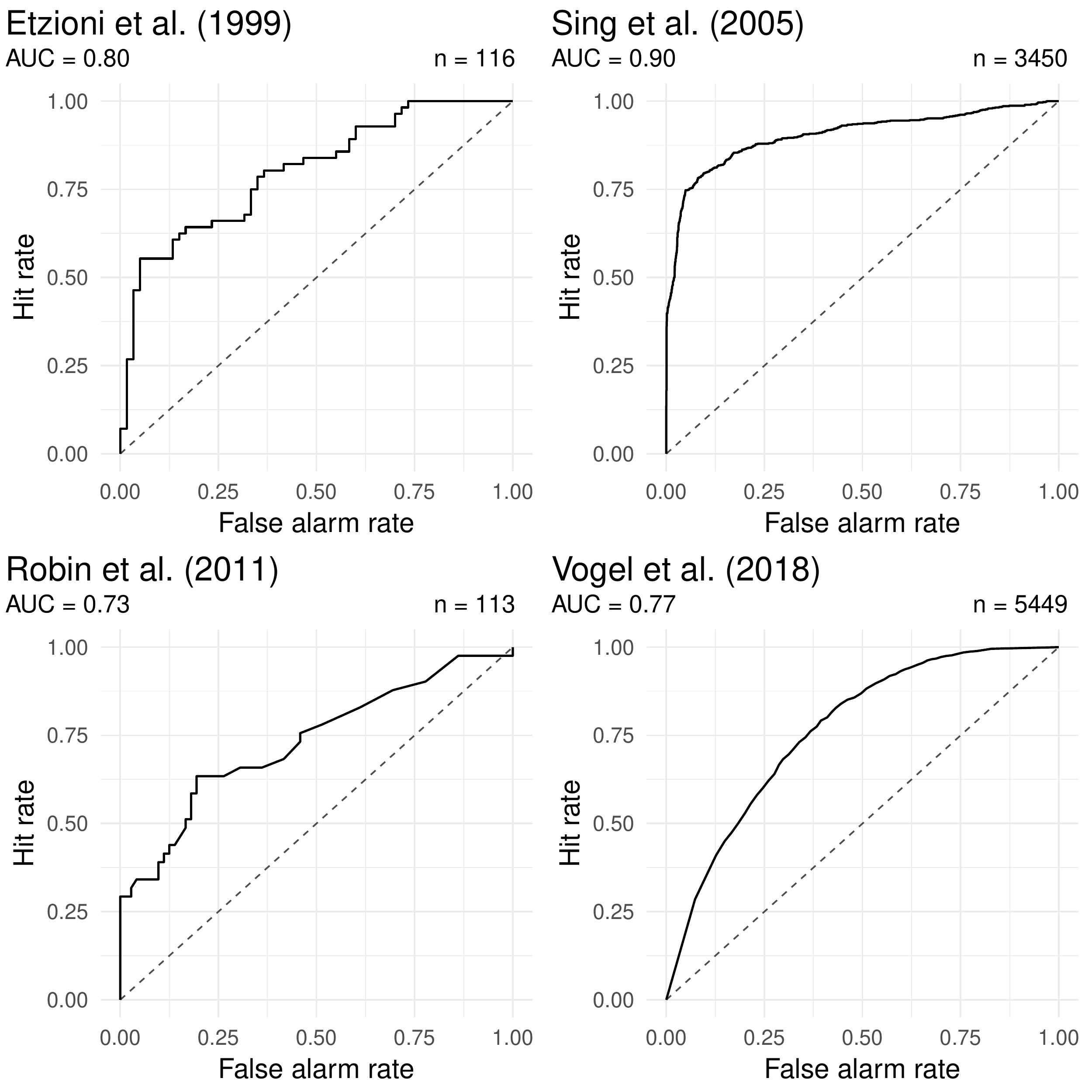}
\caption{Examples of empirical ROC curves. \label{fig:data}
}

\end{figure}

The remainder of the paper is organized as follows.
Section~\ref{sec:theory} establishes some fundamental theoretical
results.  We formalize the distinction between raw ROC diagnostics and
ROC curves, elucidate the special role of concavity in the
interpretation and modelling of ROC curves, and demonstrate an
equivalence between ROC curves and CDFs.  In Section~\ref{sec:models}
we introduce the flexible yet parsimonious two-parameter \betat model,
which uses the CDFs of \betat distributions to model ROC curves, and
we discuss estimation and testing based in this context, including
both asymptotic and Monte Carlo based approaches.  In particular, we
derive the large sample distribution of the minimum distance estimator
in general parametric settings and specialize to the \betat family and
the classical binormal model.  Section~\ref{sec:empirical} returns to
our empirical examples, of which we present detailed analyses, with
the \betat family and natural three- and four-parameter extensions
that allow for straight edges in the ROC curve fitting better than the
binormal model, particularly under the constraint of concavity.  The
paper closes with a discussion in Section~\ref{sec:discussion}.
Proofs of a more technical character are deferred to
Appendix~\ref{app}.

\section{Fundamental properties of ROC curves}  \label{sec:theory}

Consider the bivariate random vector $(X,Y)$ where $X$ is a
real-valued predictor, {\em covariate}, {\em feature}, or {\em
  marker}, and $Y$ is the binary response.  We refer to the joint
distribution of $(X,Y)$ as $\myQ$.  Let $\pi_1 = \myQ(Y = 1) \in
(0,1)$ and $\pi_0 = 1 - \pi_1 = \myQ(Y = 0)$, and let $F_1(x) = \myQ(
X \leq x \, | \, Y = 1)$, $F_0(x) = \myQ( X \leq x \, | \, Y = 0)$,
and
\[
F(x) = \myQ( X \leq x) = \pi_0 \hsp F_0(x) + \pi_1 \hsp F_1(x)
\]
denote the conditional and marginal cumulative distribution functions
(CDFs) of $X$, respectively.  Furthermore, we let $F_0(x-) = \lim_{x'
  \uparrow x} F_0(x')$.

We use column vectors to denote points in the Euclidean plane, and
given any $(a,b)' \in \real^2$ we write $(a,b)_{(1)}' = a$ and
$(a,b)_{(2)}' = b$ for the respective coordinate projections.

\subsection{Raw ROC diagnostics and ROC curves}  \label{sec:ROC}

In this common setting ROC diagnostics concern the points of the form
$(\f(x),\h(x))'$, where $\f(x) = 1 - F_0(x)$ is the {\em false alarm
  rate}\/ and $\h(x) = 1 - F_1(x)$ the {\em hit rate}\/ at the
threshold value $x \in \real$.  Formally, the {\em raw ROC
  diagnostic}\/ for the random vector $(X,Y)$ and the bivariate
distribution $\myQ$ is the point set
\begin{equation}  \label{eq:R*} 
R^* = 
\left\{ \begin{pmatrix} 1 - F_0(x) \\ 1 - F_1(x) \end{pmatrix} : x \in \real \right\} 
\end{equation} 
within the unit square.  Clearly, the bivariate distribution $\myQ$ of
$(X,Y)$ is characterized by $F_0$, $F_1$, and any of the two marginal
distributions.  In contrast, the raw ROC diagnostic along with a
single marginal does not characterize $\myQ$, due to the well known
invariance of ROC diagnostics under strictly increasing
transformations of $X$ and shifts in the prevalence of the binary
outcome (Fawcett 2006).  However, the raw ROC diagnostic along with
both marginal distributions determines $\myQ$.

\begin{theorem}  \label{thm:char.raw}
The joint distribution\/ $\myQ$ of\/ $(X,Y)$ is characterized by the
raw ROC diagnostic and the marginal distributions of\/ $X$ and\/ $Y$.
\end{theorem} 

\begin{proof} 
The mapping $g : [0,1]^2 \to [0,1]$ defined by
\[
(a,b)' 
\mapsto (1 - a) \hsp \pi_0 + (1 - b) \hsp \pi_1 
\]
induces a bijection between the raw ROC diagnostic $R^*$ and the range
of $F$.  Therefore, it suffices to note that $\myQ(X \leq x, Y \leq y)
= 0$ for $y < 0$,
\begin{align*}
\myQ(X \leq x, Y \leq y) 
& =  F_0(x) \, \pi_0 \\
& = F(x) - (1 - \h(x)) \, \pi_1 \\
& = F(x) - (1 - g_{(2)}^{-1}(F(x))) \, \pi_1 
\end{align*}
for $y \in [0,1)$, and $\myQ(X \leq x, Y \leq y) = F(x)$ for $y \geq
  1$.
\end{proof} 

\begin{table}[t]

\caption{Ordered marker values $x_1 < x_2 < \cdots < x_7$, binary
  observations, and false alarm rate (FAR) and hite rate (HR) at the
  respective threshold for the example in
  Figure~\ref{fig:example}. \label{tab:example}}

\medskip
\centering
\small

\begin{tabular}{|c|ccccccccc|}
\hline
X & $<x_1$ & $x_1$ & $x_2$ & $x_3$ & $x_4$ & $x_5$ & $x_6$ & $x_7$ & $> x_7$ \\ 
\hline
Y & & 0 & 1 & 0, 0 & 0, 0, 1 & 0, 1, 1 & 1 & 1 & \\
\hline
$\f \times 6$ & 6 & 5 & 5 & 3 & 1 & 0 & 0 & 0 & 0 \\
$\h \times 6$ & 6 & 6 & 5 & 5 & 4 & 2 & 1 & 0 & 0 \\
\hline
\end{tabular}

\smallskip

\end{table}

\begin{figure}[t]

\centering 

\includegraphics[width = 0.90 \textwidth]{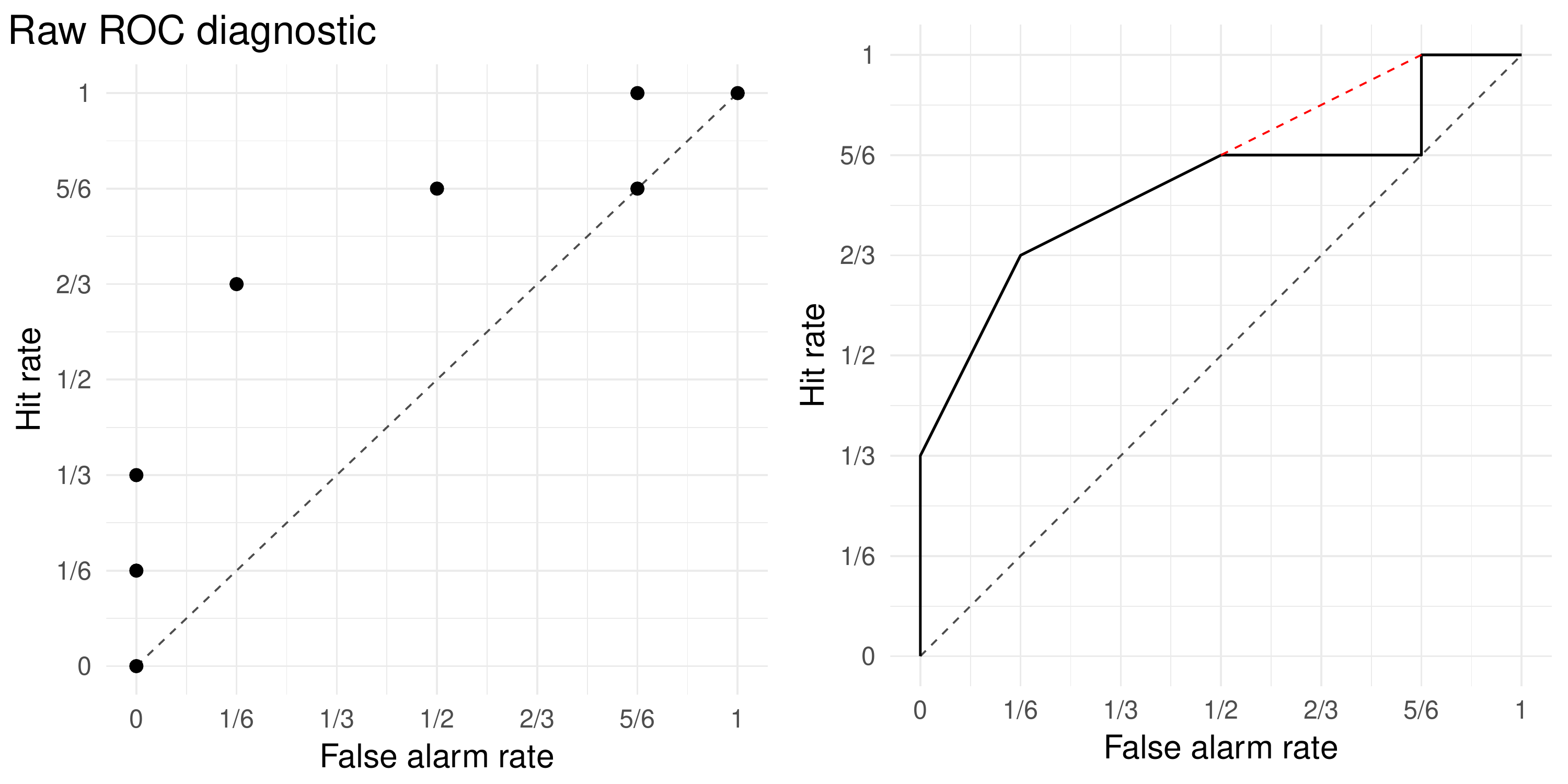}
\vspace{-2.5mm}
\caption{Raw ROC diagnostic and corresponding empirical ROC curve for
  the marker in Table~\ref{tab:example}. The broken red line completes
  the concave hull of the empirical ROC curve. \label{fig:example}}

\end{figure}

Briefly, a ROC curve is obtained from the raw ROC diagnostic by linear
interpolation.  Formally, the {\em full ROC diagnostic}\/ or {\em ROC
  curve}\/ is the point set
\begin{equation}  \label{eq:R} 
R = 
\left\{ \begin{pmatrix} 0 \\ 0 \end{pmatrix} \right\} 
\cup R^* \cup \left\{ L_x : x \in \real \right\} \cup
\left\{ \begin{pmatrix} 1 \\ 1 \end{pmatrix} \right\}
\end{equation} 
within the unit square, where 
\[
L_x = \left\{ \alpha \begin{pmatrix} 1 - F_0(x-) \\ 1 - F_1(x-) \end{pmatrix}  
+ (1-\alpha) \begin{pmatrix} 1 - F_0(x) \\ 1 - F_1(x) \end{pmatrix} : \alpha \in [0,1] \right\}
\]
is a possibly degenerate, nondecreasing line segment.  The choice of
linear interpolation to complete the raw ROC diagnostic into the ROC
curve \eqref{eq:R} is natural and persuasive, as the line segment
$L_x$ represents randomized combinations of the classifiers associated
with its end points.  In particular, linear interpolation allows for a
fair and direct comparison between continuous, discrete, and ordinal
features.  Empirical ROC curves based on samples, as illustrated in
Figure~\ref{fig:data}, fit this framework, as they arise in the
special case where $\myQ$ is an empirical measure.  We illustrate the
transition from the raw ROC diagnostic to the ROC curve in
Figure~\ref{fig:example} using the toy data set from
Table~\ref{tab:example}, where there are twelve observations and seven
unique marker values.

The raw ROC diagnostic can be recovered from the ROC curve and the two
marginal distributions, as the mapping $g$ in the proof of
Theorem~\ref{thm:char.raw} induces a bijection between the raw ROC
diagnostic and the range of $F$ that can be expressed in terms of
$\pi_1$ and $\pi_0$.  From this simple fact the following result is
immediate.

\begin{corollary}  \label{thm:char.curve}
The joint distribution\/ $\myQ$ of\/ $(X,Y)$ is characterized by the
ROC curve and the marginal distributions of\/ $X$ and\/ $Y$.
\end{corollary} 

In this sense, ROC curves and raw ROC diagnostics assume roles similar
to those of copulas (e.g., Nelsen 2006), with the difference that ROC
characteristics are defined in terms of conditional distributions,
whereas copulas operate on marginal distributions.

Given a ROC curve $R$, an obvious task is to find CDFs $F_0$ and $F_1$
that realize $R$.  For a particularly simple and appealing
construction, let $F_0$ be the CDF of the uniform distribution on the
unit interval, and take $F_1$ to be $\FNI$, defined as $\FNI(x) = 0$
for $x \leq 0$,
\begin{equation}  \label{eq:FNI} 
\FNI(x) = 1 - R_+(1-x) \quad \textrm{ for } \quad x \in (0,1),  
\end{equation}
and $\FNI(x) = 1$ for $x \geq 1$, where the function $R_+ : (0,1) \to
[0,1]$ is induced by the ROC curve at hand, in that
\[
R_+(x) = \inf \left\{ b : (a,b)' \in R, \: a \geq x \right\}. 
\]
In anticipation of its repeated use in subsequent sections, we refer
to this specific realization of a ROC curve $R$, in which $F_0$ is
standard uniform and $F_1$ is taken to be $\FNI$ in
\eqref{eq:FNI}, as the {\em natural identification}.

Remarkably, the natural identification applies even when the feature
$X$ is discrete or ordinal.  Nevertheless, the statistical models and
methods that we introduce in Section~\ref{sec:models} target the case
of a continuous marker or feature.

\subsection{Concave ROC curves}  \label{sec:concave}

We proceed to elucidate the critical role of concavity in the
interpretation and modelling of ROC curves.\footnote{Again,
  terminologies differ between communities.  In machine learning,
  concave ROC curves are typically referred to as {\em convex}\/
  (e.g., Fawcett 2006), whereas the psychological and biomedical
  literatures call them {\em proper}\/ (Egan 1975, Section 2.6, Zhou
  et al.~2011, Section 2.7.3).  The usage in this paper is in
  accordance with well established, commonly used terminology in the
  mathematical sciences.}  Its significance is well known and has been
alluded to in monographs, such as in Egan (1975, p.~35), Pepe (2003,
p.~71), and Zhou et al.~(2011, p.~40).  Nevertheless, we are unaware
of any rigorous treatment in the extant literature.  To address this
omission, we distinguish and analyse regular and discrete settings.
Unified treatments are feasible but considerably technical, and we
leave them to future work.

In the {\em regular setting}\/ we suppose that $F_1$ and $F_0$ have
continuous, strictly positive Lebesgue densities $f_1$ and $f_0$ in
the interior of an interval, which is their common support.  For every
$x$ in the interior of the support, we can define the {\em likelihood
  ratio},
\[
\LR(x) = \frac{f_1(x)}{f_0(x)}, 
\]
and the {\em conditional event probability},
\[
\CEP(x) = \myQ(Y=1 \, |  \, X=x) = \frac{\pi_1 \hsp f_1(x)}{\pi_0 \hsp f_0(x) + \pi_1 \hsp f_1(x)}.
\]
We demonstrate the equivalence of the following three conditions:
\begin{itemize} 
\item[(a)] The ROC curve is concave.
\item[(b)] The likelihood ratio is nondecreasing. 
\item[(c)] The conditional event probability is nondecreasing. 
\end{itemize} 

\begin{theorem}  \label{thm:concave.r} 
In the regular setting statements (a), (b), and (c) are equivalent.
\end{theorem} 

\begin{proof}
In the regular setting the ROC curve can be identified with a function
$R : [0,1] \to [0,1]$, where $R(p)$ is defined as in \eqref{eq:R(p)}
for $p \in (0,1)$.  If the ROC curve is concave then clearly the
function $R$ is concave as well, and so its derivative $R'(p)$ is
nonincreasing in $p \in (0,1)$.  However, the slope $R'(p)$ equals the
likelihood ratio $\LR(x)$ at a certain value $x$ that decreases with
$p$, which establishes the equivalence of (a) and (b).  Furthermore,
\begin{align*}
\LR(x) = \frac{\pi_0}{\pi_1} \, \frac{\CEP(x)}{1 - \CEP(x)}, 
\end{align*}
and the function $c \mapsto c/(1-c)$ is nondecreasing in $c \in
(0,1)$, which yields the equivalence of (b) and (c).
\end{proof}

Next we consider the {\em discrete setting}\/ in which the support of
the feature $X$ is a finite or countably infinite set.  This setting
includes, but is not limited to, the case of empirical ROC curves, as
illustrated in Figure~\ref{fig:data}.  For every $x$ in the discrete
support of $X$, we can define the {\em likelihood ratio},
\[
\LR(x) = \begin{cases}
\myQ(X = x \, | \, Y = 1) \left/ \: \myQ(X = x \, | \, Y = 0) \right. & 
 \text{if } \quad \myQ(X = x \, | \, Y = 0) > 0, \\
\infty, & \text{if } \quad \myQ(X = x \, | \, Y = 0) = 0, \\ \end{cases}
\]
and the {\em conditional event probability},
\[
\CEP(x) = \myQ(Y=1 \, |  \, X=x).
\]
In Appendix~\ref{app:concave.d} we prove the following direct analogue
of Theorem~\ref{thm:concave.r}.

\begin{theorem}  \label{thm:concave.d} 
In the discrete setting statements (a), (b), and (c) are equivalent.
\end{theorem}

The critical role of concavity in the interpretation and modelling of
ROC curves stems from the monotonicity condition (c) on the
conditional event probability, which is at the very heart of the
approach and needs to be invoked to justify the construction of just
any raw ROC diagnostic or ROC curve.  In the medical literature Hilden
(1991) notes that ``some authors do seem to overlook the concavity
problem'' and Pesce et al.~(2010) argue that ``direct use of a
decision variable'' with a non-concave ROC curve ``must be considered
irrational'' and ``unethical when applied to medical decisions''.
Similar considerations apply in the vast majority of applications of
ROC curves.

Fortunately, there are straightforward ways of restricting attention
to concave ROC curves and the associated classifiers.  Generally,
randomization can be used to generate classifiers with concave ROC
curves from features with non-concave ones (Fawcett 2006, Pesce et
al.~2010).  The regular setting serves to supply theoretical models
that can be fit to empirical ROC curves, such as the classical
binormal model or our new beta model, and the parameters in these
models can be restricted suitably to guarantee concavity, as we
discuss in Section~\ref{sec:models}.  Empirical ROC curves typically
fail to be concave, as illustrated in Figure~\ref{fig:data}.  However,
they can readily be morphed into their concave hull, by subjecting the
marker or feature at hand to the pool-adjacent violators (PAV: Ayer et
al.~1955, de Leeuw et al.~2009) algorithm, thereby converting it into
an isotonic, calibrated probabilistic classifier (Lloyd 2002, Fawcett
and Niculescu-Mizil 2007).  For example, for the toy data in
Table~\ref{tab:data} the PAV algorithm assigns the conditional event
probability $p_1 = 0$ to $x_1$, the value $p_2 = \frac{1}{3}$ to $x_2,
x_3$, and $x_4$, the value $p_3 = \frac{2}{3}$ to $x_5$, and the value
$p_4 = 1$ to $x_6$ and $x_7$.  The ROC curve for this isotonic and
calibrated probabilistic classifier is the concave hull of the ROC
curve for the original marker, as shown in Figure~\ref{fig:example}.

\subsection{An equivalence between ROC curves and probability measures}  \label{sec:equivalence}

We move on to provide concise and practically relevant
characterizations of ROC curves, both with and without the critical
condition of concavity.

\begin{theorem}  \label{thm:R} 
There is a one-to-one correspondence between ROC curves and
probability measures on the unit interval.  In particular, the natural
identification induces a bijection between the class of the ROC curves
and the class of the CDFs of probability measures on the unit
interval.
\end{theorem} 

\begin{proof}
Given a ROC curve, we remove any vertical line segments, except for
the respective upper endpoints, to yield the CDF of a probability
measure on the unit interval.  Conversely, given the CDF of a
probability measure on the unit interval, we interpolate vertically at
any jump points to obtain a ROC curve.  This mapping is a bijection,
and save for the symmetries in \eqref{eq:FNI} it is realized by the
natural identification.
\end{proof}

We say that a curve $C$ in the Euclidean plane is nondecreasing if
$a_0 \leq a_1$ is equivalent to $b_0 \leq b_1$ for points $(a_0,
b_0)', (a_1, b_1)' \in C$.  The following result is immediate.

\begin{corollary}  \label{cor:R.curve} 
The ROC curves are the nondecreasing curves in the unit square that
connect the points\/ $(0,0)'$ and\/ $(1,1)'$.
\end{corollary} 

We now state characterizations under the constraint of strict
concavity.  Analogous results hold under the slightly weaker
assumption of concavity.

\begin{theorem}  \label{thm:R.concave} 
There is a one-to-one correspondence between strictly concave ROC
curves and probability measures with strictly decreasing Lebesgue
densities on the unit interval, which is induced by the natural
identification.
\end{theorem} 

\begin{corollary}  \label{cor:R.curve.concave} 
The strictly concave ROC curves are in one-to-one correspondence to
the strictly concave functions\/ $R$ on the unit interval with\/ $R(0)
= 0$ and\/ $R(1) = 1$.
\end{corollary} 

Turning to methodological and applied considerations, these results
support a shift of paradigms in the statistical modelling of ROC
curves.  In extant practice, the emphasis is on modelling the
conditional distributions $F_0$ and $F_1$, such as in the ubiquitous
binormal model. Our results suggest a subtle but important change of
perspective, in that ROC modelling can be approached as an exercise in
curve fitting,\footnote{While curve fitting approaches have been
  advocated before, such as by Swets (1986, p.~104, his approach (b)),
  they lacked theoretical support.}  with any nondecreasing curve that
connects $(0,0)'$ to $(1,1)'$ being a permissible candidate, and
parametric families of CDFs on the unit interval offering particularly
attractive models, including but not limited to the \betat family that
we introduce in the next section.

\section{Parametric models, estimation, and testing}  \label{sec:models}

The binormal model is by far the most frequently used parametric model
and ``plays a central role in ROC analysis'' (Pepe 2003, p.~81).
Specifically, the binormal model assumes that $F_1$ and $F_0$ are
Gaussian with means $\mu_1 \geq \mu_0$ and strictly positive variances
$\sigma_0^2$ and $\sigma_1^2$, respectively.  We are in the regular
setting of Section~\ref{sec:concave}, and the resulting ROC curve is
represented by the function $R : [0,1] \to [0,1]$ with $R(0) = 0$,
\begin{equation}  \label{eq:binormal}
R(p) = \Phi( \mu + \sigma \hsp \Phi^{-1}(p)) \quad \mbox{for} \quad p \in (0,1),
\end{equation}
and $R(1) = 1$, where $\Phi$ is the CDF of the standard normal
distribution, $\mu = (\mu_1 - \mu_0)/\sigma_1 \geq 0$ is a scaled
difference in expectations, and $\sigma = \sigma_0/ \sigma_1$ is the
ratio of the respective standard deviations.  The respective area
under the curve is 
\[
\AUC(\mu,\sigma) = \Phi \! \left( \frac{\mu}{\sqrt{1 + \sigma^2}} \right) \! .
\]
For an illustration of binormal ROC curves see the left-hand panel of
Figure~\ref{fig:parametric}.  It is well known that a binormal ROC
curve is concave only if $\sigma = 1$ or equivalently if $F_0$ and
$F_1$ differ in location only.  These curves are necessarily symmetric
with respect to the anti-diagonal in the unit square, thereby strongly
inhibiting their flexibility, as illustrated in the left-hand panel of
Figure~\ref{fig:parametric.concave}.

\begin{figure}[p]

\centering 

\includegraphics[width = 0.95 \textwidth]{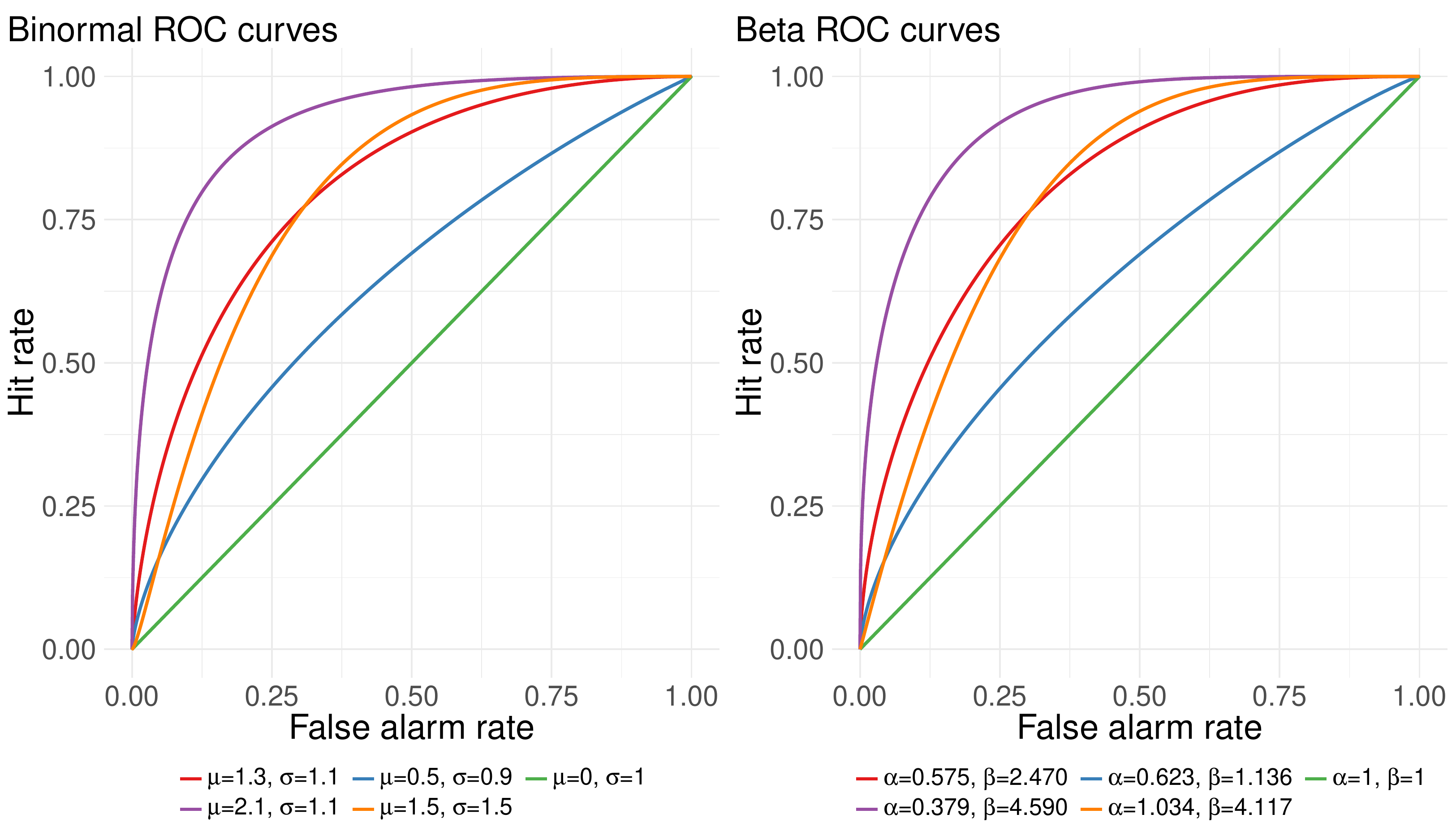}
\caption{Members of the (left) binormal family and (right) \betat
  family of ROC curves. The parameter values for the beta curves have
  been chosen to match the overall shape of the same-color binormal
  ROC curve.
\label{fig:parametric}}

\bigskip
\bigskip

\includegraphics[width = 0.95 \textwidth]{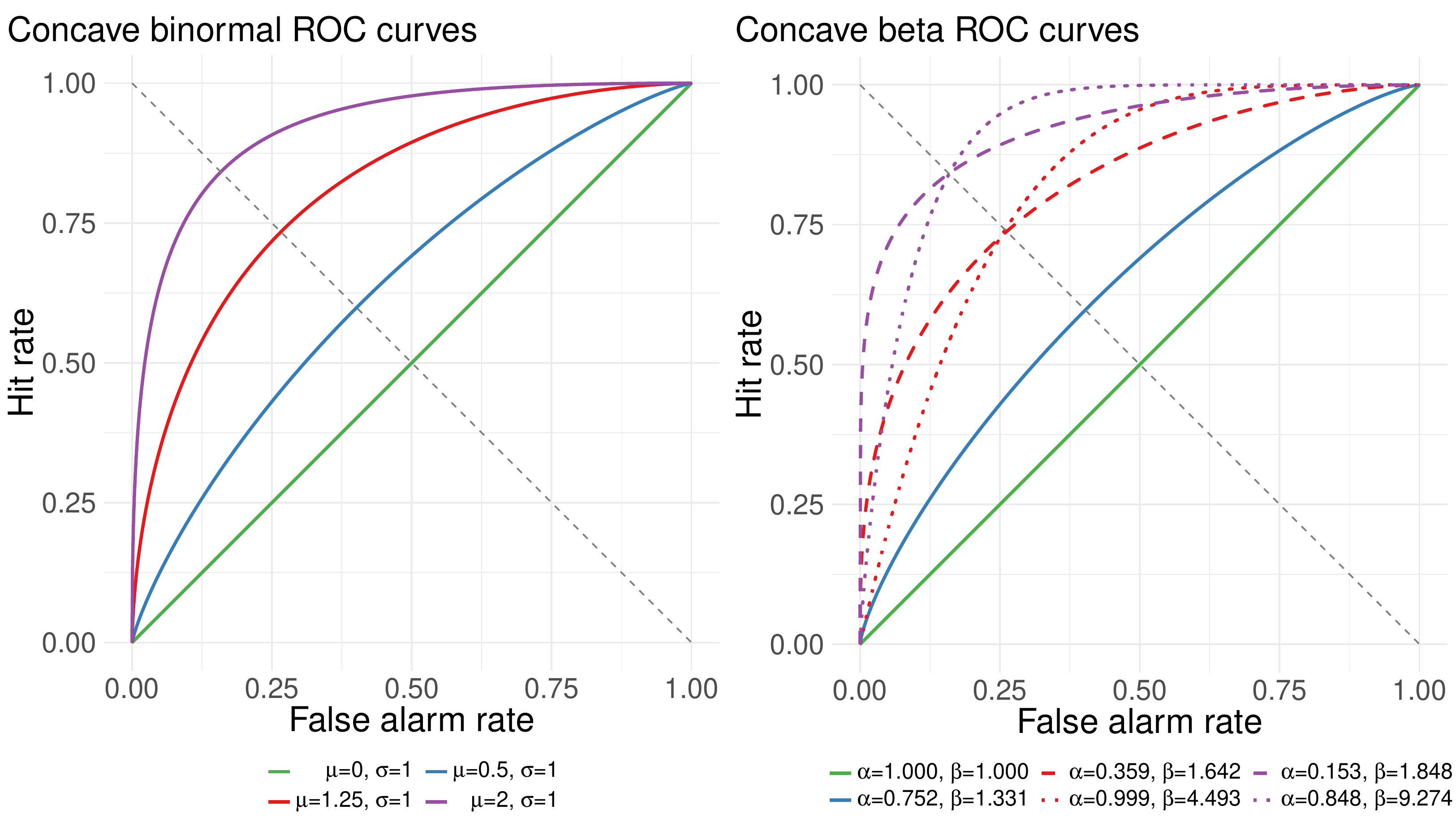}
\caption{Concave members of the (left) binormal family and (right)
  \betat family of ROC curves. The parameter values for the beta
  curves have been chosen to match the value of the same-color
  binormal ROC curve at the
  anti-diagonal.  \label{fig:parametric.concave}}

\end{figure}

\subsection{The \betat model}  \label{sec:beta}

Motivated and supported by the characterization theorems of
Section~\ref{sec:theory}, we propose a curve fitting approach to the
statistical modelling of ROC curves, with the two-parameter family of
the cumulative distribution functions (CDFs) of \betat distributions
being a particularly attractive model.  Specifically, consider the
\betat family with ROC curves represented by the function
\begin{equation} \label{eq:beta} 
R(p) = B_{\alpha,\beta}(p) = \int_0^p b_{\alpha,\beta}(q) \dd q
\quad \mbox{for} \quad p \in [0,1],
\end{equation} 
where $b_{\alpha,\beta}(q) \propto q^{\alpha - 1} (1-q)^{\beta - 1}$
is the density of the beta distribution with parameter values $\alpha
> 0$ and $\beta > 0$.  As illustrated in Figure~\ref{fig:beta} and
shown in Appendix~\ref{app:beta}, a \betat ROC curve is concave if
$\alpha \leq 1$ and $\beta \geq 2 - \alpha$, and its \AUC \ value is
\[
\AUC(\alpha,\beta) = \frac{\beta}{\alpha + \beta}.  
\]
In the limit as $\beta \to \infty$ we obtain the perfect ROC curve
with straight edges from $(0,0)'$ to $(0,1)'$ and $(1,1)'$,
corresponding to a complete separation of the supports of $F_1$ and
$F_0$.  While the requirement of a concave ROC curve is restrictive,
the condition is much less stringent than for the binormal family,
where it constrains the admissible parameter space to a single
dimension.  The adaptability of the \betat family is illustrated in
Figure~\ref{fig:parametric}, where we see that members of the beta
family can match the shape of binormal ROC curves, and in
Figure~\ref{fig:parametric.concave}, where the gain in flexibility
under the critical constraint of concavity is evident.  

\begin{figure}[t]

\centering 

\includegraphics[width = 0.40 \textwidth]{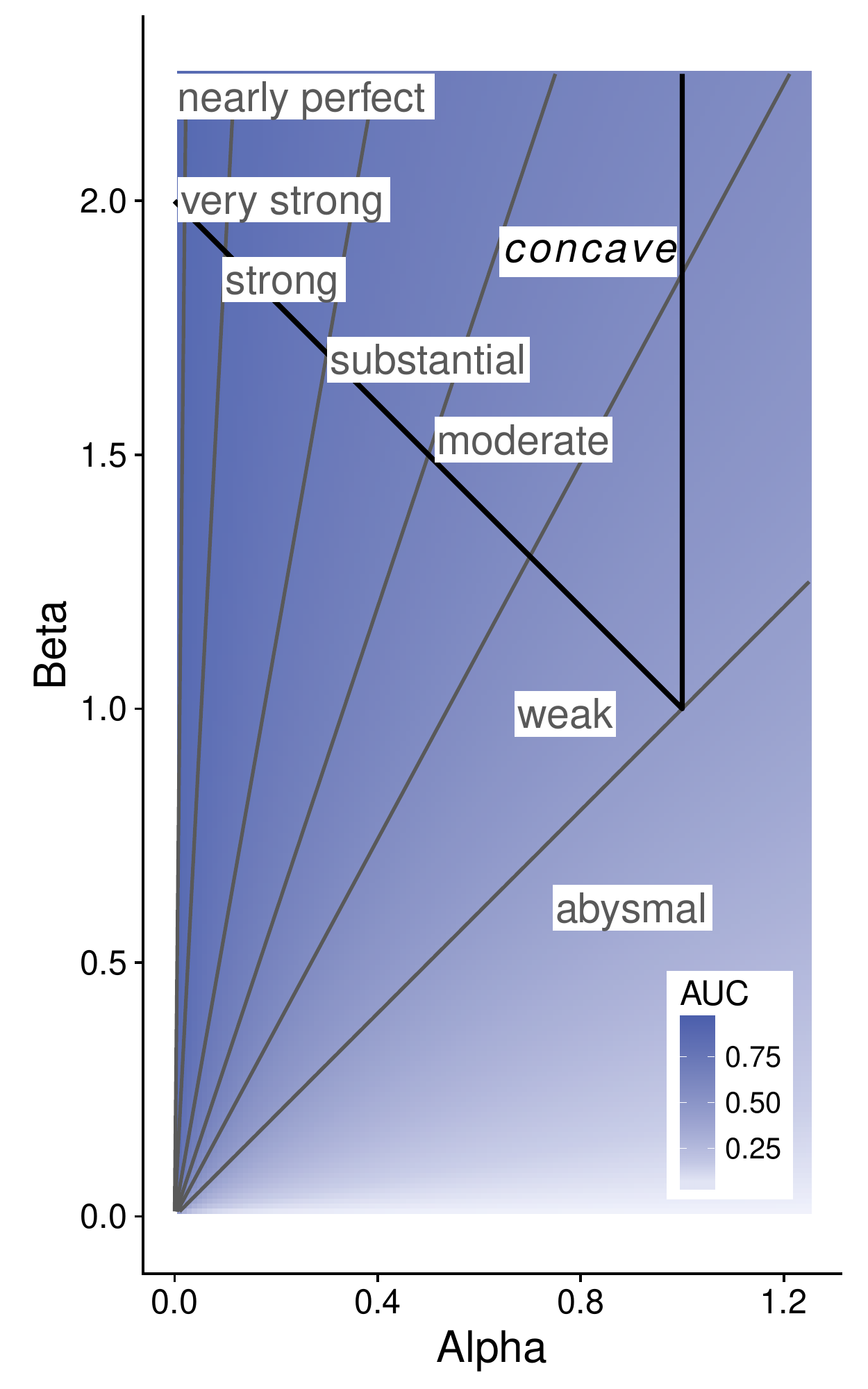}
\caption{AUC value for the \betat family of ROC curves.  The isolines
  correspond to the terminology for predictor strength introduced in
  Table~\ref{tab:AUC}. Parameter combinations above the thick line
  yield concave ROC curves.
\label{fig:beta}}

\end{figure}

The \betat family nests the time-honored one-parameter power model
(Egan et al.~1961, Swets 1986) that arises in the special case when
$\beta = 1$.  While the classical derivation of the power model does
not readily generalize, our theoretical results justify the use of the
two-parameter \betat family.  If even further flexibility is desired,
mixtures of \betat CDFs, i.e., functions of the form
\[
R_n(p) = \sum_{k=1}^n w_k B_{\alpha_k,\beta_k}(p) \quad \mbox{for} \quad p \in [0,1],
\]
where $w_1, \ldots, w_n \geq 0$ with $w_1 + \cdots + w_n = 1$,
$\alpha_1, \ldots, \alpha_k > 0$, and $\beta_1, \ldots, \beta_k > 0$,
approximate any regular ROC curve to any desired accuracy, as
demonstrated by the following result.  Recall from
Section~\ref{sec:concave} that in the regular setting the ROC curve
can be identified with the function $R$ in \eqref{eq:R(p)}, where
$F_1$ and $F_0$ have continuous, strictly positive Lebesgue densities
$f_1$ and $f_0$ in the interior of an interval, which is their common
support.  A ROC curve is {\em regular}\/ if it arises in this way and
{\em strongly regular}\/ if furthermore the derivative $R'$ is
bounded.

\begin{theorem}  \label{thm:approx} 
For every strongly regular ROC curve\/ $R$ there is a sequence of
mixtures of beta CDFs that converges uniformly to\/ $R$.
\end{theorem} 

The proof of this result relies on Bernstein's probabilistic approach
to the Weierstrass theorem (Levasseur 1984) and is deferred to
Appendix \ref{app:beta}.

\subsection{Minimum distance estimation}  \label{sec:MDE}

For the parametric estimation of ROC curves for continuous markers
various methods have been proposed, including maximum likelihood
(Dorfman and Alf 1969, Metz et al.~1998, Zou and Hall 2000),
approaches based on generalized linear models (Pepe 2000), and minimum
distance estimation (Hsieh and Turnbull 1996), as reviewed at book
length by Pepe (2003), Krzanowski and Hand (2009), and Zhou et
al.~(2011).  Maximum likelihood techniques face a conceptual
challenge, in that ROC curves do not determine the joint distribution
of the marker and the binary event.  Here we pursue the minimum
distance estimator, which is much in line with our curve fitting
approach.

We assume a parametric model in the regular setting of
Section~\ref{sec:concave}, where now the ROC curve depends on a
parameter $\theta \in \Theta \subseteq \real^k$.  Specifically, we
suppose that for each $\theta \in \Theta$ the ROC curve is represented
by a smooth function
\[
R(p; \theta) = 1 - F_{1,\theta}(F_{0,\theta}^{-1}(1-p)) 
\quad \mbox{for} \quad p \in (0,1),
\]
where $F_{1,\theta}$ and $F_{0,\theta}$ admit continuous, strictly
positive densities $f_{1,\theta}$ and $f_{0,\theta}$ in the interior
of an interval, which is their common support.  We also require that
the true parameter value $\theta_0$ is in the interior of the
parameter space $\Theta$, where the derivative
\[
R'(p;\theta) 
= \frac{\partial R(p;\theta)}{\partial p} 
= \frac{f_{1,\theta}(F_{0,\theta}^{-1}(1-p))}{f_{0,\theta}(F_{0,\theta}^{-1}(1-p))} 
\]
exists and is finite for $p \in (0,1)$, and where the partial
derivative $R_{(i)}(p;\theta)$ of $R(p;\theta)$ with respect to
component $i$ of the parameter vector $\theta = (\theta_1, \ldots,
\theta_k)'$ exists and is continuous for $i = 1, \ldots, k$ and $p \in
(0,1)$.

We adopt the asymptotic scenario of Hsieh and Turnbull (1996), where
at sample size $n$ there are $n_0$ and $n_1 = n - n_0$ independent
draws from $F_{0,\theta}$ and $F_{1,\theta}$ with corresponding binary
outcomes of zero and one, respectively, where $\lambda_n = n_0/n_1$
converges to some $\lambda \in (0,1)$ as $n \to \infty$.  For $\theta
\in \Theta$ we define the difference process
\[
\xi_n(p;\theta) = \hat{R}_n(p) - R(p;\theta), 
\] 
where the function $\hat{R}_n(p)$ represents the empirical ROC curve.
The minimum distance estimator $\hat{\theta}_n = (\hat{\theta}_1,
\dots, \hat{\theta}_k)_n'$ then satisfies
\[
\| \xi_n(\cdot\hsp;\hat{\theta}_n) \| 
= {\textstyle \min_{\theta \in \Theta}} \| \xi_n(\cdot\hsp;\theta) \|, 
\]
where $\| \xi_n(\cdot\hsp;\theta) \| = (\int_0^1 \xi_n(p;\theta)^2 \dd
p)^{1/2}$ is the standard $L_2$-norm.  If $n$ is large, $\hat\theta_n$
exists and is unique with probability approaching one (Millar 1984),
and so we follow the extant literature in ignoring issues of existence
and uniqueness.

The minimum distance estimator has a multivariate normal limit
distribution in this setting, as suggested by the asymptotic result of
Hsieh and Turnbull (1996) that under the usual $\sqrt{n}$ scaling the
difference process $\xi_n(p; \theta)$ has limit
\begin{equation}  \label{eq:W}
W(p;\theta) = \sqrt{\lambda} \, B_1(R(p;\theta)) + R'(p;\theta) \hsp B_2(p)
\end{equation}
at $\theta = \theta_0$, where $B_1$ and $B_2$ are independent copies
of a Brownian bridge.  In Appendix~\ref{app:MDE} we review the
specifics of the convergence to the limit process \eqref{eq:W} and
combine results of Millar (1984) and Hsieh and Turnbull (1996) to show
the following result.

\begin{theorem}  \label{thm:MDE}
In the above setting the minimum distance estimator\/ $\hat{\theta}_n$
satisfies
\begin{equation}  \label{eq:limit}
\sqrt{n} \, (\hat{\theta}_n - \theta_0) \to \cN(0, C^{-1}AC^{-1}) 
\end{equation}
as $n \to \infty$, where the matrices\/ $A$ and\/ $C$ have entries
\begin{equation}  \label{eq:AC}   
A_{ij} = \int_0^1 \int_0^1 R_{(i)}(s;\theta_0) \hsp K(s,t;\theta_0) \hsp R_{(j)}(t;\theta_0) \dd s \dd t, 
\quad 
C_{ij} = \int_0^1 R_{(i)}(s;\theta_0) \hsp R_{(j)}(s;\theta_0) \dd s
\end{equation}
for\/ $i, j = 1, \ldots, k$, respectively, and where
\begin{align}  \label{eq:K}
K(s,t;\theta_0) \notag
& = \lambda \hsp (\min \{ R(s;\theta_0), R(t;\theta_0) \} - R(s;\theta_0) R(t;\theta_0)) \\
& \hspace{40mm} + \hsp R'(s;\theta_0) R'(t;\theta_0) \hsp (\min\{s,t\} - st). 
\end{align}
is the covariance function of the process\/ $W(p;\theta)$ in
\eqref{eq:W} at\/ $\theta = \theta_0$.
\end{theorem}

\begin{corollary}  \label{cor:AUC}
In the above setting,
\begin{equation}  \label{eq:limit.AUC}
\sqrt{n} \, (\AUC(\hat{\theta}_n) - \AUC(\theta_0)) \to \cN(0, GC^{-1}AC^{-1}G'), 
\end{equation}
where $G$ is the gradient of the mapping\/ $\theta \mapsto
\AUC(\theta)$ at\/ $\theta = \theta_0$.
\end{corollary}

Both the binormal and the beta model satisfy the assumptions for these
results, which allow for asymptotic inference about the model
parameters and the AUC, by plugging in $\hat\theta_n$ for $\theta_0$
in the expressions for the asymptotic covariances.  For the binormal
model \eqref{eq:binormal} we have $\theta = (\mu,\sigma)$,
$R_{(\mu)}(p;\theta) = \varphi(\mu + \sigma \hsp \Phi^{-1}(p))$,
$R_{(\sigma)}(p;\theta) = \Phi^{-1}(p) \hsp \varphi(\mu + \sigma \hsp
\Phi^{-1}(p))$, and $R'(p;\theta) = \sigma \hsp \varphi(\mu + \sigma
\hsp \Phi^{-1}(p)) / \varphi(\Phi^{-1}(p))$, where $\varphi$ is the
standard normal density, so that the integrals in \eqref{eq:AC} can
readily be evaluated numerically.  The gradient in
\eqref{eq:limit.AUC} equals $G = (\varphi(\mu_0/\sqrt{1 + \sigma_0^2})
/ (1 + \sigma_0^2)^{1/2}, - \mu_0 \hsp \sigma_0 \hsp
\varphi(\mu_0/\sqrt{1 + \sigma_0^2}) / (1 + \sigma_0^2)^{3/2})$.
Hsieh and Turnbull (1996) consider binormal ordinal dominance curves,
which interchange the roles of $F_1$ and $F_0$ relative to ROC curves,
and after reparameterization we recover their results.  However, the
formula for the covariance function $K(s,t;\theta)$ in the first
displayed equation on page 39 in Hsieh and Turnbull (1996) is
incompatible with our equation \eqref{eq:K} and incorrect, as it is
independent of $s$ and $t$ and therefore constant.  Under the \betat
model \eqref{eq:beta} we have $\theta = (\alpha,\beta)$ and
$R'(p;\theta) = b_{\alpha,\beta}(p)$.  While closed form expressions
for the partial derivatives of $R(p;\theta)$ with respect to $\alpha$
and $\beta$ exist, they are difficult to evaluate, and we approximate
them with finite differences.  The gradient in \eqref{eq:limit.AUC}
equals $G = (- \beta_0/(\alpha_0 + \beta_0)^2, \alpha_0/(\alpha_0 +
\beta_0)^2)$.

\subsection{Testing goodness-of-fit and other hypotheses}  \label{sec:tests}

We move on to discuss testing.  A natural hypothesis to be addressed
is whether a given parametric model fits the data at hand.  In
contrast to existing methods that are based on AUC and focus on the
binormal model (Zou et al.~2005), we propose a simple Monte Carlo test
that applies to any parametric model $\cC$.  For example, $\cC$ could
be the full binormal, the concave binormal, the full beta, or the
concave beta family.  While we describe the procedure for minimum
distance estimates and the $L_2$-distance, it applies equally to other
estimates and other distance measures.

Given a dataset of size $n$ with $n_0$ instances where the binary
outcome is zero and $n_1 = n - n_0$ instances where it is one, our
goodness-of-fit test proceeds as follows.  We use the notation of
Section \ref{sec:MDE} and denote the number of Monte Carlo replicates
by $M$.
\begin{enumerate} 
\item 
Fit a model from class $\cC$ to the empirical ROC curve for the data
at hand, to yield the minimum distance estimate $\theta_{\rm data}$.
Compute $d_{\rm data}$ as the $L_2$-distance between the fitted and
the empirical ROC curve.
\item
For $m = 1, \ldots, M$,
  \begin{enumerate}
  \item
  draw a sample of size $n$ under $\theta_{\rm data}$, with $n_0$ and
  $n_1$ instances from $F_{0,\theta_{\rm data}}$ and $F_{1,\theta_{\rm
      data}}$ and associated binary outcomes of zero and one,
  respectively,
  \item
  fit a model from class $\cC$ to the empirical ROC curve, to yield
  the minimum distance estimate, and
  \item
  compute $d_m$ as the $L_2$-distance between the fitted and the
  empirical ROC curve.
  \end{enumerate} 
\item 
Find a $p$-value based on the rank of $d_{\rm data}$ when pooled with
$d_1, \ldots, d_M$.  Specifically, $p = (\# \{ i = 1, \ldots, M :
d_{\rm data} \leq d_i \} + 1) / (M+1)$.
\end{enumerate}
Under the null hypothesis of the ROC curve being generated by a random
sample within class $\cC$ the Monte Carlo $p$-value is very nearly
uniformly distributed, as is readily seen in simulation experiments
(not reported on here).

Parametric tests of the equality of ROC curves and AUC values can be
based on the limit distributions in Theorem~\ref{thm:MDE} and
Corollary \ref{cor:AUC} in the usual way.  Under an identifiable model
the hypothesis of two ROC curves being equal is the same as the
hypothesis of the respective parameters being the same.  Therefore,
the limit in \eqref{eq:limit} allows for a customary chi square test
of the equality of ROC curves from independent samples, based on the
squared norm of the normalized difference between the two estimates of
the parameter vector, as proposed by Metz and Kronman (1980) in the
case of maximum likelihood estimates under the binormal model.
Similarly, the limit in \eqref{eq:limit.AUC} justifies a $z$-test for
the equality of the AUC values, based on the normalized difference
between the two parametric estimates of the AUC.  We illustrate the
use of these tests in the subsequent section and provide software for
their implementation in the case of independent samples.  For paired,
dependent samples, correlations between the estimates need to be
accounted for, a task to be addressed in future work.  As an
alternative, nonparametric tests have been developed in the extant
literature (Hanley and McNeil 1983, DeLong et al.~1988, Venkatraman
and Begg 1996, Venkatraman 2000, Mason and Graham 2002).

\section{Empirical examples}  \label{sec:empirical}

We return to the empirical ROC curves in Figure \ref{fig:data} and
present basic information about the underlying datasets in Table
\ref{tab:data}.  In the dataset from Etzioni et al.~(1999), the
negative logarithm of the ratio of free to total prostate-specific
antigen (PSA) two years prior to diagnosis in serum from patients
later found to have prostate cancer is compared to age-matched
controls.  The datasets from Sing et al.~(2005, Figure 1a) and Robin
et al.~(2011, Figure 1) are prominent examples in the widely used {\tt
  ROCR} and {\tt pROC} packages in \textsf{R}.  They concern a score
from a linear support vector machine (SVM) trained to predict the
usage of HIV coreceptors, and the S100$\beta$ biomarker as it relates
to a binary clinical outcome, respectively.  The dataset from Vogel et
al.~(2018, Figure~6d) considers probability of precipitation forecasts
from the European Centre for Medium-Range Weather Forecasts (ECMWF)
numerical weather prediction (NWP) ensemble system (Molteni et
al.~1996) for the binary event of precipitation occurrence within the
next 24 hours at meteorological stations in the West Sahel region in
northern tropical Africa.

\begin{table}[t]

\caption{Basic information about the datasets and minimum distance
  estimates under the unrestricted and concave binormal and \betat
  models for the ROC curves in Figures~\ref{fig:data} and
  \ref{fig:data.parametric}.  Fit is in terms of the $L_2$-distance to
  the empirical ROC curve, and the $p$-value is from the
  goodness-of-fit test of Section \ref{sec:tests}. \label{tab:data}}

\bigskip
\centering 

\footnotesize

\begin{tabular}{lcccc}
\hline
Dataset \rule{0mm}{4mm} & Etzioni       & Sing          & Robin         & Vogel \\
                         & et al.~(1999) & et al.~(2005) & et al.~(2011) & et al.~(2018) \\
\hline
Binary outcome \rule{0mm}{4mm} & prostate cancer & coreceptor usage & clinical outcome & precipitation \\
Feature & antigen ratio & SVM predictor  & S100$\beta$ concentr. & NWP forecast \\
Sample size & 116 & 3450 & 113 & 5449 \\
\hline
Binormal model \rule{0mm}{4mm} & & & & \\
unrestricted $(\mu, \sigma)$ & (1.05, 0.78)  & (1.58, 0.65) & (0.75, 0.72) & (1.13, 1.22) \\
\hspt fit       & 0.043 & 0.019 & 0.033 & 0.008 \\
\hspt $p$-value & 0.106 & 0.001 & 0.561 & 0.032 \\
concave $(\mu, \sigma)$ & (1.22, 1.00) & (2.05, 1.00) & (0.91, 1.00) & (0.99, 1.00) \\
\hspt fit       & 0.056 & 0.039 & 0.060 & 0.031 \\ 
\hspt $p$-value & 0.138 & 0.001 & 0.147 & 0.001 \\
\hline
Beta model \rule{0mm}{4mm} & & & & \\
unrestricted $(\alpha, \beta)$ & (0.34, 1.32) & (0.15, 1.44) & (0.36, 0.96) & (0.79, 2.57) \\
\hspt fit       & 0.042 & 0.023 & 0.032 & 0.006 \\ 
\hspt $p$-value & 0.117 & 0.001 & 0.620 & 0.187 \\
concave $(\alpha, \beta)$ & (0.38, 1.62) & (0.17, 1.83) & (0.51, 1.49) & (0.79, 2.57) \\
\hspt fit       & 0.045 & 0.025 & 0.050 & 0.006 \\ 
\hspt $p$-value & 0.196 & 0.001 & 0.204 & 0.171 \\
\hline
\end{tabular}

\end{table}

\begin{figure}[t]

\centering 

\includegraphics[width = 0.89 \textwidth]{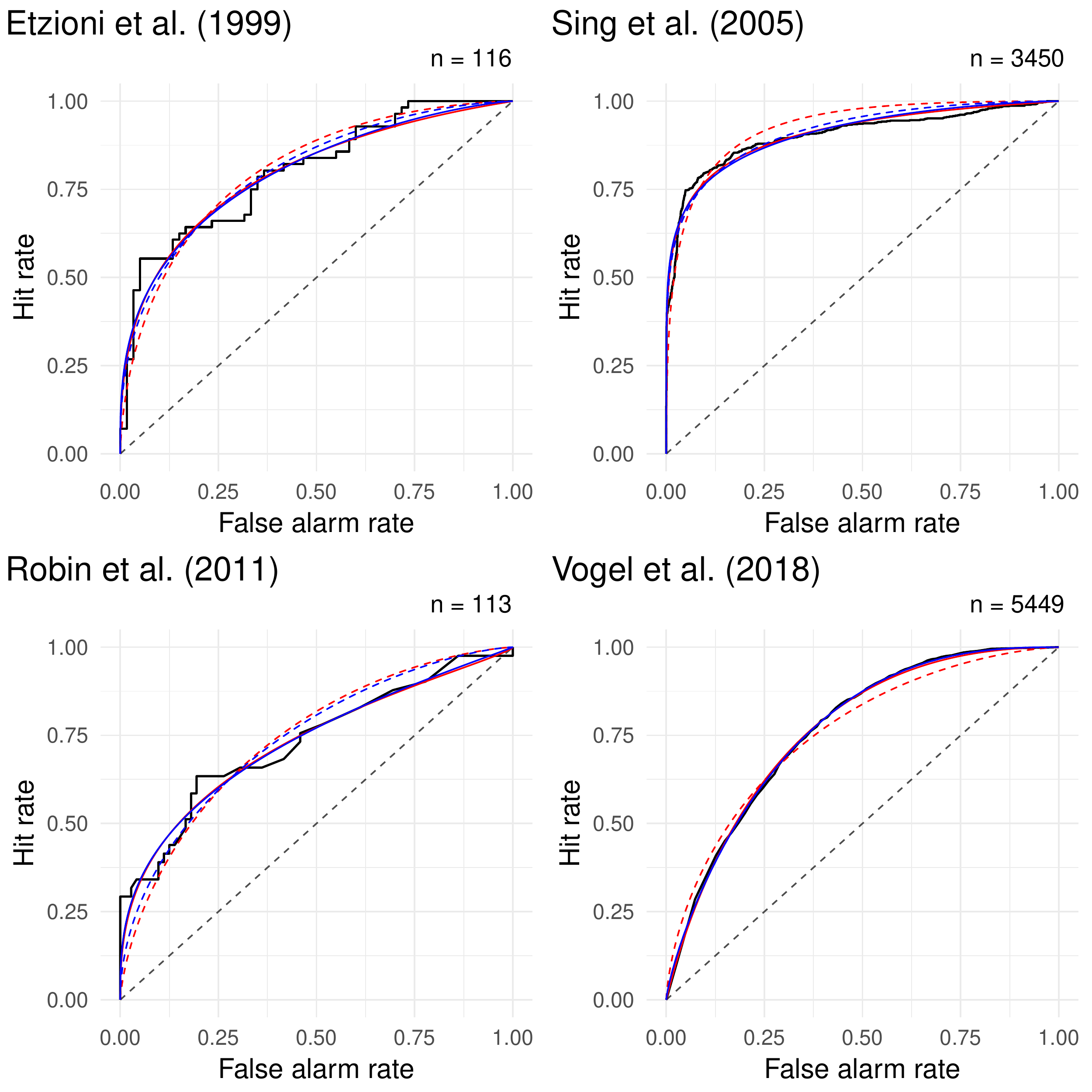}
\caption{Fitted binormal (red) and \betat (blue) ROC curves in the
  unrestricted (solid) and concave (dashed) case for the datasets from
  Figure~\ref{fig:data} and
  Table~\ref{tab:data}. \label{fig:data.parametric}}

\end{figure}

Figure~\ref{fig:data.parametric} shows binormal and beta ROC curves
fitted to the empirical ROC curves, both in the unrestricted case and
under the constraint of concavity.  The respective unrestricted and
restricted minimum distance estimates, the fit in terms of the
$L_2$-distance to the empirical ROC curve, and the $p$-value from the
goodness-of-fit test in Section \ref{sec:tests} with $M = 999$ Monte
Carlo replicates, are given in Table \ref{tab:data}.  In the
unrestricted case, the binormal and \betat fits are visually nearly
indistinguishable.  The fitted binormal ROC curves fail to be concave
and change markedly when concavity is enforced.  For the \betat ROC
curves, the differences between restricted and unrestricted fits are
less pronounced, and in the example from Vogel et al.~(2018) the
unrestricted fit is concave.  Generally, in the constrained case the
improvement in the fit under the more flexibel \betat model as
compared to the classical binormal model is substantial.

\begin{figure}[t]

\centering 

\includegraphics[width = 0.45 \textwidth]{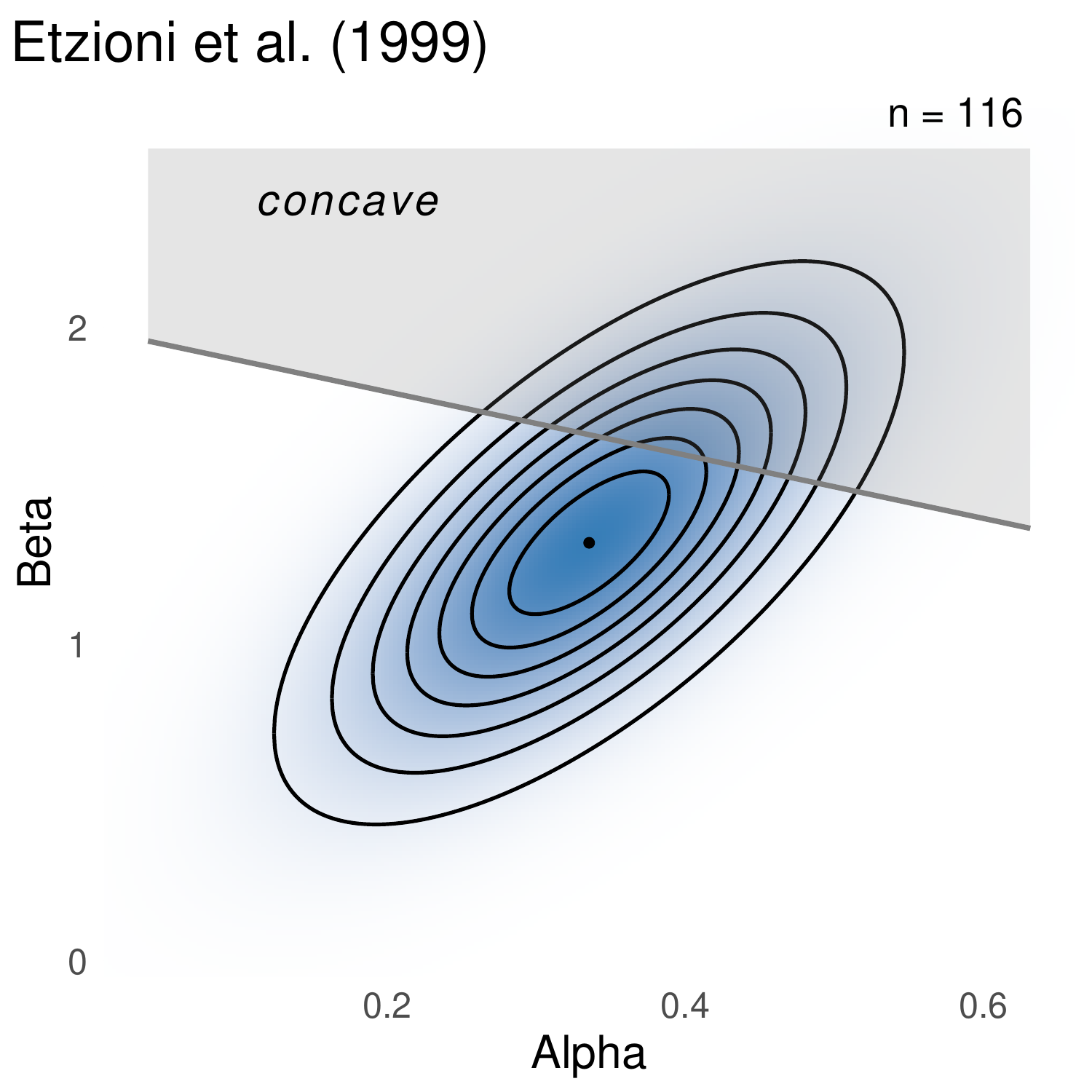}
\caption{Asymptotic inference under the unrestricted \betat model for
  the data from Etzioni et al.~(1999).  The confidence ellipses are
  at level $1/8, 2/8, \ldots, 7/8$, respectively. \label{fig:Etzioni}}

\end{figure}

The theoretical results in Section \ref{sec:MDE} allow for asymptotic
inference about the model parameters.  We illustrate this in Figure
\ref{fig:Etzioni} for the unrestricted \betat fit for the dataset from
Etzioni et al.~(1999).  In addition to showing confidence ellipsoids,
we indicate and separate concave and non-concave fits.  If we seek to
complement the minimum distance estimate with pointwise confidence
bands for the ROC curve, we can sample from the inferred distribution
for the model parameters and display the envelope of the respective
ROC curves, as examplified in Figure \ref{fig:Vogel} below.

\begin{figure}[t]

\centering 

\includegraphics[width = 0.80 \textwidth]{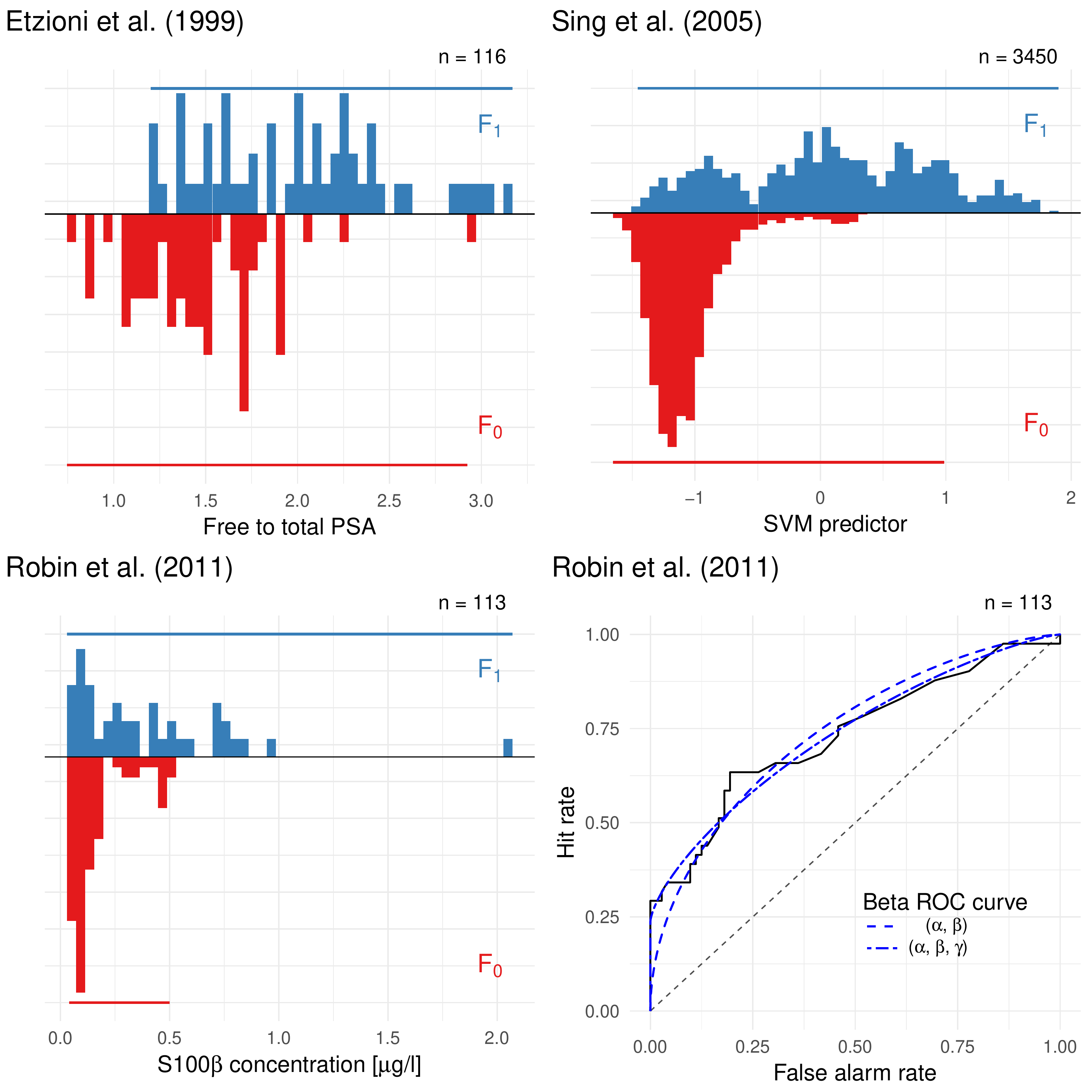}
\caption{Histograms for the conditional distributions for data in
  Table \ref{tab:data}, and concave two- and three-parameter \betat
  ROC curves fit to the data from Robin et al.~(2011).  The horizontal
  lines in the histograms extend to the convex hull of the respective
  support.  \label{fig:straight}}

\end{figure}

A closer look at the empirical ROC curves for the biomedical data from
Etzioni et al.~(1999), Sing et al.~(2005) and Robin et al.~(2011) in
Figures~\ref{fig:data} and \ref{fig:data.parametric} reveals a
striking commonality, in that the curves show vertical and/or
horizontal straight edges.  From the definition of the raw ROC
characteristic \eqref{eq:R*} it is evident that vertical straight
edges correspond to marker values that may allow for deterministic
class attribution, as illustrated in the back-to-back histograms in
Figure \ref{fig:straight}.  Importantly, straight edges might convey
critical information from a subject matter perspective, such as in
medical diagnoses, where straight edges in ROC curves correspond to
particularly high or low marker values that might identify individuals
as healthy or diseased beyond doubt.

Under the \betat family the statistical modeling of straight edges is
straightforward.  Specifically, we can generalize the two-parameter
model \eqref{eq:beta} to a four-parameter beta family, where
\begin{equation}  \label{eq:4p}
R(p) = \gamma + (1 - \gamma) \hsp B_{\alpha, \beta} \! \left( \frac{p}{\delta} \right) 
\quad \mbox{for} \quad p \in (0,1],
\end{equation}
which allows for a vertical straight edge that connects the coordinate
origin $(0,0)'$ to the point $(0,\gamma)'$, and a horizontal straight
edge that connects the points $(\delta,1)'$ and $(1,1)'$ within the
ROC curve.  Three-parameter subfamilies with a single type of straight
edge arise if we fix $\delta = 1$ and let $\gamma \in [0,1]$ vary, or
fix $\gamma = 0$ and consider $\delta \in (0,1]$, respectively.  While
  the subfamily with $\delta = 1$ being fixed has a direct analogue
  under the binormal model, there is no natural way of adapting the
  subfamily with $\gamma = 0$ being fixed or the four-parameter family
  in \eqref{eq:4p} to the binormal case.

To be clear, we do {\em not}\/ advocate uncritical routine use of the
four-parameter family in \eqref{eq:4p} and the respective
three-parameter subfamilies.  However, we {\em do}\/ recommend that in
any specific application researchers check for straight edges in
empirical ROC curves, and assess on the basis of substantive expertise
whether or not they ought to be modeled.  Visual tools such as the
back-to-back histograms for the conditional distributions in Figure
\ref{fig:straight} can assist in this assessment.  For illustration,
the back-to-back histograms might suggest that we fit the
three-parameter model with $\gamma = 0$ being fixed to the data from
Etzioni et al.~(1999) and the three-parameter model with $\delta = 1$
being fixed to the data from Sing et al.~(2005) and Robin et
al.~(2011).  While in the first two cases the three-parameter fits are
nearly identical to the fits under the two-parameter \betat model, the
three-parameter extension yields a substantially improved fit for the
data from Robin et al.~(2011), as illustrated in the lower right panel
of Figure \ref{fig:straight}.  The constrained minimum distance
estimate for $(\alpha, \beta, \gamma)$ is $(0.70, 1.30, 0.24)$ with
$L_2$-distance 0.029 to the empirical ROC curve.  For comparison,
under the two-parameter concave \betat model the estimate for
$(\alpha, \beta)$ is $(0.51, 1.49)$ with $L_2$-distance 0.050.

\begin{figure}[t]

\centering 

\includegraphics[width = 0.465 \textwidth]{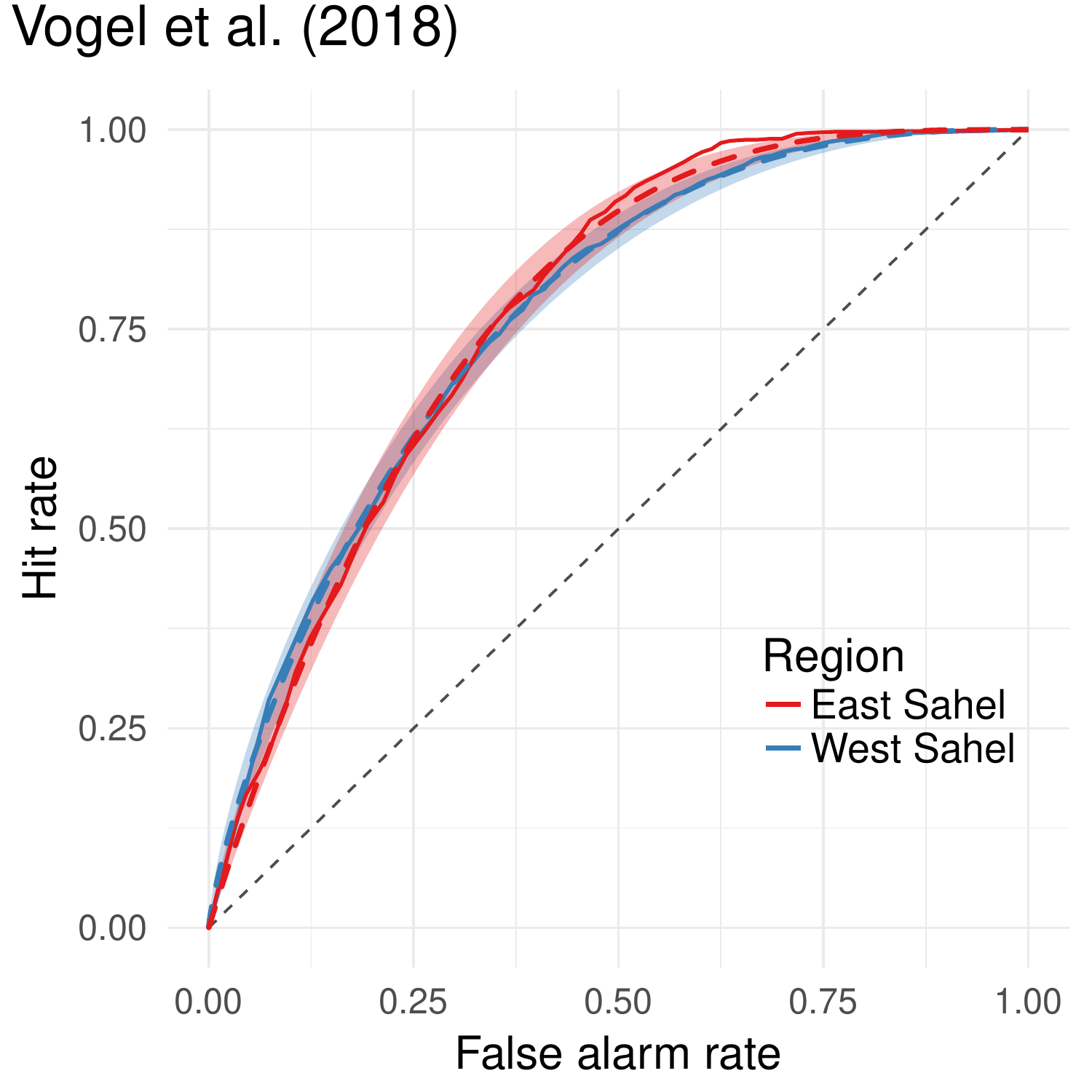}
\caption{Empirical ROC curve (solid), concave \betat fit (dashed), and
  associated pointwise 95\% confidence band for data from Vogel et
  al.~(2018) on probability of precipitation forecasts over West and
  East Sahel in northern tropical Africa.
  \label{fig:Vogel}}

\end{figure}

Finally, we take another look at the meteorological data from Vogel et
al.~(2018).  Here it is obvious from the scientific context in weather
prediction that the above three- and four-parameter extensions are
irrelevant.  While the data introduced and analyzed in
Table~\ref{tab:data} and Figures~\ref{fig:data} and
\ref{fig:data.parametric} concern probability of precipitation
forecasts over the West Sahel region, Vogel et al.~(2018) consider the
East Sahel region as well.  The respective empirical ROC curves are
shown in Figure \ref{fig:Vogel} along with the constrained
two-parameter \betat fit and parametric 95\% pointwise confidence
bands.  The $p$-value for the goodness-of-fit test of Section
\ref{sec:tests} is 0.168 for West Sahel and 0.057 for East Sahel.  Our
parametric tests for equality of AUC values and ROC curves yield
$p$-values of 0.633 and 0.015, whereas the nonparametric tests of
DeLong et al.~(1988) and Venkatraman (2000) result in $p$-values of
0.616 and 0.089, respectively.

\section{Discussion}  \label{sec:discussion}

ROC curves have been used extensively to evaluate the potential
predictive value of covariates, features, or markers in binary
problems in a multitude of scientific disciplines.  Their appeal stems
from attractive and desirable properties in this context, which
include the straightforward interpretation of ROC curves in terms of
attainable operating conditions (i.e., hit and false alarm rates),
their invariance under strictly increasing transformations of the
feature and shifts in prevalence, and the interpretation of AUC as the
probability of a marker value drawn from $F_1$ being higher than a
value drawn independently from $F_0$.  We emphasize that ROC curves
and AUC values ``should be regarded as a measure of potential rather
than actual skill'' (Kharin and Zwiers 2003, p.~4148) tailored to
serve the purposes of variable selection and feature screening across
all types of ordinal, discrete, and continuous predictor
variables.\footnote{ROC curves and AUC values have limitations when
  they are used to assess the actual skill of probability forecasts,
  as they ignore the critical requirement of calibration (Wilks 2011,
  p.~346).  For evaluating the {\em actual}\/ skill and value of
  probabilistic classifiers, proper scoring rules (Gneiting and
  Raftery 2007) are a preferred tool, notably in the form of Murphy
  diagrams (Ehm et al.~2016).  For a direct comparison of ROC curves
  and Murphy diagrams and a respective discussion in the context of
  probability forecasts see Figure 6 and Sections 2.c and 4 of Vogel
  et al.~(2018).}

Despite their ubiquitous use, our understanding of fundamental
properties of ROC curves has been incomplete.  The theoretical results
in Section \ref{sec:theory} establish an equivalence between ROC
curves and the CDFs of probability measures on the unit interval,
which motivates and justifies our curve fitting approach to the
modelling of ROC curves.  Concave fits are preferred, if not
essential, as they characterize the predictor variables with
nondecreasing likelihood ratios and nondecreasing conditional event
probabilities.  The \betat family \eqref{eq:beta} provides a
particularly attractive parametric model.  As compared to the
classical binormal model the beta family is considerably more flexible
under the constraint of concavity, and it embeds naturally into the
four-parameter model \eqref{eq:4p} that allows for straight edges in
the ROC curve.  If further flexibility is sought, mixtures of beta
CDFs can be fitted.  With a view toward nonparametric alternatives,
one might model (minus) the second derivative of a regular ROC curve,
which is nonnegative under the concavity constraint.

For estimation we focus on the minimum distance approach.  In the
regular setting, where features are continuous, minimum distance
estimates and associated parametric estimates of the AUC value are
asymptotically normal.  Goodness of fit and other hypotheses can be
tested for based on these methods and results.  In view of the
critical role of concavity for the interpretation of ROC curves, an
interesting and relevant question is whether or not one should subject
features to the PAV algorithm (de Leeuw et al.~2009) prior to fitting
a concave model.  The PAV algorithm morphs the empirical ROC curve
into the respective concave hull, and its use for data pre-processing
in other types of shape-constrained estimation problems has been
examined by Mammen (1991).  The derivation of the large sample
distributions in Section \ref{sec:MDE} is based on empirical process
theory (Shorack and Wellner 2009), and it depends on the Gaussian
limit in \eqref{eq:W}, which does not apply under ordinal or discrete
features nor when ROC curves have straight edges.  We leave the
derivation of large sample distributions for minimum distance
estimates in these cases as well as adaptations to covariate- and
time-dependent settings (Etzioni et al.~1999, Heagerty et al.~2000) to
future work.

Datasets and code in \tR (\tR Core Team 2017) for replicating our
results and implementing the proposed estimators and tests is
available from Peter Vogel.  We are working towards an \tR package
tentatively named {\tt betaROC}, to be released shortly.

\section*{Acknowledgments}

The research leading to these results has been accomplished within
project C2 ``Prediction of wet and dry periods of the West African
monsoon'' of the Transregional Collaborative Research Center SFB / TRR
165 ``Waves to Weather'' funded by the German Science Foundation
(DFG).  Tilmann Gneiting thanks the Klaus Tschira Foundation for
generous support, and we are grateful to Alexander Jordan, Johannes
Resin, and our meteorological collaborators Andreas H.~Fink, Peter
Knippertz, and Andreas Schlueter for a wealth of comments and
continuous encouragement.

{

\small

\appendix

\section{Appendix: Proofs}  \label{app}

\subsection{Concave ROC curves: The discrete setting}  \label{app:concave.d} 

In proving Theorem~\ref{thm:concave.d} we may assume that the support
of $X$ is a finite or countably infinite, ordered set of two or more
points $x_i$, indexed by consecutive integers such that $x_i < x_j$ if
$i < j$.  In the case of a finite set we assume that it is of
cardinality at least 2 and adapt the arguments in obvious ways to
account for boundary effects.

\begin{lemma}  \label{le:concave.d}
Any of the statements in Theorem~\ref{thm:concave.d} implies that either
\begin{itemize}
\item[(i)] $\myQ(X = x_i \, | \, Y = 0) > 0$ for all\/ $i$, or 
\item[(ii)] there exists an index value\/ $i^*$ such that\/ $\myQ(X =
  x_i \, | \, Y = 0) = 0$ for all\/ $i \geq i^*$ and\/ $\myQ(X = x_i
  \, | \, Y = 0) > 0$ for all\/ $i < i^*$.
\end{itemize}
\end{lemma}

\begin{proof}
If any of the statements in Theorem~\ref{thm:concave.d} hold and
condition (i) is violated, there exists an index $i$ such that $\myQ(X
= x_i \, | \, Y = 0) = 0 \}$.  Then $\CEP(x_i) = 1$, $\LR(x_i) =
\infty$, and the ROC curve has a vertical straight edge, which
contradicts statements (c), (b), and (a), respectively, unless
condition (ii) is satisfied and the straight edge is at the
origin.
\end{proof} 

\begin{proof}[Proof of Theorem~\ref{thm:concave.d}] 
In view of Lemma~\ref{le:concave.d}, it suffices to show the
equivalence of the statements in Theorem~\ref{thm:concave.d} for
indices $i$ with $\myQ(X = x_i \, | \, Y = 0) > 0$.  The fact that  
\[
\LR(x_i) = \frac{\pi_0}{\pi_1} \, \frac{\CEP(x_i)}{1 - \CEP(x_i)} 
\]
along with the monotonicity of the function $c \mapsto c/(1-c)$
establishes the equivalence of (b) and (c).  Furthermore, the
relationship
\begin{align*}
\LR(x_i) 
& = \frac{\myQ(X = x_i \, | \, Y = 1)}{\myQ(X = x_i \, | \, Y = 0)} \\
& = \frac{\myQ(X > x_{i-1} \, | \, Y = 1) - \myQ(X > x_i \, | \, Y = 1)}
       {\myQ(X > x_{i-1} \, | \, Y = 0) - \myQ(X > x_i \, | \, Y = 0)} \\
& = \frac{\h(x_{i-1}) - \h(x_{i})}{\f(x_{i-1}) - \f(x_{i})} 
\end{align*}
implies that  
\begin{align*} 
\LR(x_{i+1}) \geq \LR(x_i) 
\; \Leftrightarrow \; 
\frac{\h(x_{i}) - \h(x_{i+1})}{\f(x_{i}) - \f(x_{i+1}} 
\geq 
\frac{\h(x_{i-1}) - \h(x_{i})}{\f(x_{i-1}) - \f(x_{i})} 
\end{align*}
and the right-hand side is equivalent to the ROC curve being concave,
thereby demonstrating the equivalence of (a) and (b).
\end{proof}

\subsection{Properties of \betat ROC curves}  \label{app:beta}

\begin{lemma}
The AUC value for the \betat ROC curve is $\beta / (\alpha + \beta)$.
\end{lemma}

\begin{proof}
We have 
\begin{align*}
\AUC(\alpha, \beta) 
= \int_0^1 B_{\alpha,\beta}(p) \dd p 
= \left[ p \hsp B_{\alpha,\beta}(p) - \frac{\alpha}{\alpha + \beta} \hsp B_{\alpha+1,\beta}(p) \right]_0^1 
= 1 - \frac{\alpha}{\alpha + \beta}, 
\end{align*}
as claimed.
\end{proof}

\begin{lemma}
The CDF of the \betat distribution is concave if and only if $\alpha
\leq 1$ and $\beta \geq 2 - \alpha$, and it is strictly concave if
furthermore $\alpha < 1$.
\end{lemma}

\begin{proof}
The density $b_{\alpha,\beta}$ of the beta distribution satisfies 
\[
b_{\alpha,\beta}'(x) = \frac{\alpha - 1 + (2 - \alpha - \beta) x}{x (1 - x)} \ b_{\alpha, \beta}(x)
\]
for $x \in (0,1)$, from which the statement is immediate. 
\end{proof}

\begin{proof}[Proof of Theorem \ref{thm:approx}]
We apply the natural identification and define $\FNI$ as in
\eqref{eq:FNI}.  Due to the assumption of strong regularity, $\FNI$
admits a density on $(0,1)$ that can be extended to a continuous
function $\fNI$ on $[0,1]$.  The arguments in Bernstein's
probabilistic proof of the Weierstrass approximation theorem
(Levasseur 1984) show that as $n \to \infty$ the sequence
\[
 m_n(q) = \frac{1}{n+1} \sum_{k=0}^n \fNI \! \left( \frac{k}{n} \right) b_{k+1,n-k+1}(q) 
\]
converges to $\fNI(q)$ uniformly in $q \in [0,1]$.  Furthermore,
\[
a_n = \int_0^1 m_n(q) \dd q \to \int_0^1 \fNI(q) \dd q = 1
\]
as $n \to \infty$, and for $n = 1, 2, \ldots$ the mapping $p \mapsto
M_n(p) = \int_0^p m_n(q) \dd q / a_n$ respresents a mixture of \betat
CDFs.  The uniform convergence of $m_n$ to $\fNI$ implies that for
every $\epsilon > 0$ there exists an $n'$ such that
\begin{align*}
|\FNI(p) - M_n(p)| 
& \leq \int_0^p \left| \fNI(q) - \frac{m_n(q)}{a_n} \right| \! \dd q \\
& \leq \int_0^p \left| \fNI(q) - \frac{\fNI(q)}{a_n} \right| \! \dd q 
          + \frac{1}{a_n} \int_0^p \left| \fNI(q) - m_n(q) \right| \! \dd q \\
& \leq \left| 1 - \frac{1}{a_n} \right| 
          + \frac{1}{a_n} \int_0^p \left| \fNI(q) - m_n(q) \right| \! \dd q 
< \epsilon
\end{align*}
for all integers $n > n'$ uniformly in $p \in [0,1]$.  The statement
of the theorem follows.
\end{proof}

\subsection{Asymptotic normality of minimum distance estimates}  \label{app:MDE} 

Here we demonstrate the asymptotic normality of the minimum distance
estimator $\hat{\theta}_n$ in the setting of Section~\ref{sec:MDE}.
In a nutshell, we apply Theorem 2.2 of Hsieh and Turnbull (1996) and
Theorem 3.6 along with the results in Section II in the fundamental
paper on minimum distance estimation by Millar (1984).  In contrast to
the results in Section 4 of Hsieh and Turnbull (1996), which concern
minimum distance estimation for the binormal model and ordinal
dominance curves, Theorem \ref{thm:MDE} applies to general parametric
families and ROC curves.

\begin{proof}[Proof of Theorem~\ref{thm:MDE}]
We are in the setting of Theorem 2.2 of Hsieh and Turnbull (1996),
according to which there exists a probability space with sequences
$(B_{1,n})$ and $(B_{2,n})$ of independent versions of Brownian
bridges such that
\begin{equation}  \label{eq:W.app}
\sqrt{n} \, \xi_n(p; \theta_0) = 
\sqrt{\lambda} \, B_{1,n}(R(p;\theta_0)) + R'(p;\theta_0) \hsp B_{2,n}(p)
+ o \! \left( n^{-1/2} (\log n)^2 \right)
\end{equation}
almost surely, and uniformly in $p$ on every interval $[a,b] \subset
(0,1)$.  We proceed to verify the regularity conditions for Theorem
3.6 of Millar (1984).  As regards the identifiability condition (3.2)
and the differentiability condition (3.5) it suffices to note that
\[
\xi_n(p; \theta) - \xi_n(p; \theta_0) = R(p;\theta_0) - R(p;\theta)
\]
is nonrandom, continuously differentiable with respect to $p$ and the
components of the parameter vector $\theta$, and independent of $n$.
The boundedness condition (3.3) is trivially satisfied and the
convergence condition (3.4) is implied by \eqref{eq:W.app}.  Finally,
we apply\footnote{We note a typographical error in eq.~(2.20) of
  Millar (1984), where the asymptotic covariance matrix is incorrectly
  specified as $C^{-1} A C$; it should read $C^{-1} A C^{-1}$
  instead.}  (2.17), (2.18), (2.19), and (2.20) in Section II of
Millar (1984) to yield \eqref{eq:limit} and \eqref{eq:AC}, where the
covariance function of the process in \eqref{eq:W} is
\begin{align*}
K(s,t;\theta) 
& = \text{Cov}(W(s;\theta), W(t;\theta)) \\
& = \lambda \, \text{Cov}(B_1(R(s;\theta)), B_1(R(t;\theta))) 
    + R'(s;\theta) R'(t;\theta) \, \text{Cov}(B_2(s), B_2(t)) \\ 
& = \lambda \hsp (\min \{ R(s;\theta), R(t;\theta) \} - R(s;\theta) R(t;\theta)) 
    + R'(s;\theta) R'(t;\theta) (\min \{ s, t \} - st),
\end{align*}
whence $K(s,t;\theta_0)$ is as stated in \eqref{eq:K}.
\end{proof} 

The asymptotic result in Corollary \ref{cor:AUC} follows in a
straightforward application of the delta method.

}

\section*{References}

\newenvironment{reflist}{\begin{list}{}{\itemsep 0mm \parsep 1mm
\listparindent -7mm \leftmargin 7mm} \item \ }{\end{list}}

\vspace{-6.5mm}
\begin{reflist}

Ayer, M., Brunk, H., Ewing, G., Reid, W.~and Silverman, E.~(1955).  An
empirical distribution function for sampling with incomplete
information.  {\em Annals of Mathematical Statistics}, {\bf 5},
641--647.

Bradley, A.~P.~(1997).  The use of the area under the ROC curve in the
evaluation of machine learning algorithms.  {\em Pattern Recognition},
{\bf 30}, 1145--1159.

DeLong, E.~R., DeLong, D.~and Clarke-Pearson, D.~L.~(1988).  Comparing
the areas under two or more correlated receiver operating
characteristic curves: A nonparametric approach.  {\em Biometrics},
{\bf 44}, 837--845.

Dorfman, D.~D.~and Alf, E.~(1969).  Maximum-likelihood estimation of
parameters of signal-detection theory and determination of confidence
intervals --- rating-method data.  {\em Journal of Mathematical
  Psychology}, {\bf 6}, 487--496.
  
Egan, J.~P.~(1975).  {\em Signal Dectection Theory and ROC Analysis}.
Academic Press, New York.

Egan, J.~P., Greenberg, G.~Z.~and Schulman, A.~I.~(1961).  Operating
characteristics, signal detectability, and the method of free
response.  {\em Journal of the Acoustical Society of America}, {\bf
  33}, 993--1007.

Ehm, W., Gneiting, T., Jordan, A.~and Kr\"uger, F.~(2016).  Of
quantiles and expectiles: Consistent scoring functions, Choquet
representations and forecast rankings (with discussion and rejoinder).
{\em Journal of the Royal Statistical Society Series B: Statistical
Methodology}, {\bf 78}, 505--562.

Etzioni, R., Pepe, M., Longton, G., Hu, C.~and Goodman, G.~(1999).
Incorporating the time dimension in receiver operating characteristic
curves: A case study of prostate cancer.  {\em Medical Decision
  Making}, {\bf 19}, 242--251.

Fawcett, T.~(2006).  An introduction to ROC analysis.  {\em Pattern
  Recognition Letters}, {\bf 27}, 861--874.

Fawcett, T.~and Niculescu-Mizil, A.~(2007).  PAV and the ROC convex
hull.  {\em Machine Learning}, {\bf 68}, 97--106.

Gneiting, T.~and Raftery, A.~E.~(2007).  Strictly proper scoring
rules, prediction, and estimation.  {\em Journal of the American
  Statistical Association}, {\bf 102}, 359--378.

Hanley, J.~A.~and McNeil, B.~J.~(1982).  The meaning and use of the
area under a receiver operating characteristic (ROC) curve.  {\em
  Radiology}, {\bf 143}, 29--36.

Hanley, J.~A.~and McNeil, B.~J.~(1983).  A method of comparing the
areas under receiver operating characteristic curves derived from the
same cases.  {\em Radiology}, {\bf 148}, 839--843.

Heagerty, P.~J., Lumley, T.~and Pepe, M.~S.~(2000).  Time-dependent
ROC curves for censored survival data and a diagnostic marker.  {\em
  Biometrics}, {\bf 56}, 337--344.

Hilden, J.~(1991).  The area under the ROC curve and its competitors.
{\em Medical Decision Making}, {\bf 11}, 95--101.

Hsieh, F.~and Turnbull, B.~W.~(1996).  Nonparametric and
semiparametric estimation of the receiver operating characteristic
curve.  {\em Annals of Statistics}, {\bf 24}, 25--40.

Kharin, V.~V.~and Zwiers, F.~(2003).  On the ROC score of probability
forecasts.  {\em Journal of Climate}, {\bf 16}, 4145--4150.

Krzanowski, W.~J.~and Hand, D.~J.~(2009).  {\em ROC Curves for
  Continuous Data.}  CRC Press, Boca Raton.

de Leeuw, J., Hornik, K.~and Mair, P.~(2009).  Isotone optimization in
\textsf{R}: Pool-Adjacent-Violators Algorithm (PAVA) and active set methods.
{\em Journal of Statistical Software}, {\bf 32 (5)}, 1--24.

Levasseur, K.~M.~(1984).  A probabilistic proof of the Weierstrass
approximation theorem.  {\em American Mathematical Monthly}, {\bf 91},
249--250.

Lloyd, C.~J.~(2002).  Estimation of a convex ROC curve.  {\em
  Statistics \& Probability Letters}, {\bf 59}, 99--111.

Mammen, E.~(1991).  Estimating a smooth monotone regression function. 
{\em Annals of Statistics}, {\bf 19}, 724--740.

Mason, S.~J.~and Graham, N.~E.~(2002).  Areas beneath the relative
operating characteristic (ROC) and relative operating levels (ROL)
curves: Statistical significance and interpretation.  {\em Quarterly
  Journal of the Royal Meteorological Society}, {\bf 128}, 2145--2166.

Metz, C.~E.~and Kronman, H.~B.~(1980).  Statistical significance tests
for binormal ROC curves.  {\em Journal of Mathematical Psychology},
{\bf 22}, 218--243.

Metz, C.~E., Herman, B.~A.~and Shen, J.-H.~(1998).  Maximum likelihood
estimation of receiver operating characteristic (ROC) curves from
continuously-distributed data.  {\em Statistics in Medicine}, {\bf
  17}, 1033--1053.

Millar, P.~W.~(1984).  A general approach to the optimality of minimum
distance estimators.  {\em Transactions of the American Mathematical
  Society}, {\bf 286}, 377--418.

Molteni, F., Buizza, R., Palmer, T.~N.~and Petroliagis, T.~(1996).
The ECMWF ensemble prediction system: Methodology and validation.
{\em Quarterly Journal of the Royal Meteorological Society}, {\bf
122}, 73--119.

Nelsen, R.~B.~(2006).  {\em An Introduction to Copulas}, second edition. 
Springer, New York. 

Pepe, M.~S.~(2000).  An interpretation for the ROC curve and inference
using GLM procedures.  {\em Biometrics}, {\bf 56}, 352--359.

Pepe, M.~S.~(2003).  {\em The Statistical Evaluation of Medical Tests
  for Classification and Prediction.}  Oxford University Press, Oxford.

Pesce, L.~L., Metz, C.~E.~and Berbaum, K.~S.~(2010).  On the convexity
of ROC curves estimated from radiological results.  {\em Academic
  Radiology}, {\bf 17}, 960--968.

\tR Core Team (2017).  \textsf{R}: A language and environment for
statistical computing.  \tR Foundation for Statistical Computing,
Vienna, Austria, \url{http://www.R-project.org/}.

Robin, X., Turck, N., Hainard, A., Tiberti, N., Lisacek, F., Sanchez,
J.-C.~and M{\"u}ller, M.~(2011).  pROC: An open-source package for \tR
and \textsf{S}$+$ to analyze and compare ROC curves.  {\em BMC
  Bioinformatics}, {\bf 12}, 77.

Shorack, G.~R.~and Wellner, J.~A.~(2009).  {\em Empirical Processes
  with Applications to Statistics.}  SIAM, Philadelphia.

Sing, T., Sander, O., Beerenwinkel, N.~and Lengauer, T.~(2005).  ROCR:
Visualizing classifier performance in \textsf{R}.  {\em
  Bioinformatics}, {\bf 21}, 3940--3941.

Swets, J.~A.~(1973).  The relative operating characteristic in
psychology.  {\em Science}, {\bf 182}, 990--1000.

Swets, J.~A.~(1986). Indices of discrimination or diagnostic
accuracy.  {\em Psychological Bulletin}, {\bf 99}, 100--117.

Swets, J.~A.~(1988).  Measuring the accuracy of diagnostic systems.
{\em Science}, {\bf 240}, 1285--1293.

Venkatraman, E.~S.~(2000).  A permutation test to compare receiver
operating characteristic curves.  {\em Biometrics}, {\bf 56},
1134--1138.

Venkatraman, E.~S.~and Begg, C.~B.~(1996).  A distribution-free
procedure for comparing receiver operating characteristic curves from
a paired experiment.  {\em Biometrika}, {\bf 83}, 835--848.

Vogel, P., Knippertz, P., Fink, A.~H., Schlueter, A.~and Gneiting,
T.~(2018).  Skill of global raw and postprocessed ensemble predictions
of rainfall over northern tropical Africa.  {\em Weather and
  Forecasting}, {\bf 33}, 369--388.

Wilks, D.~S.~(2011).  {\em Statistical Methods in the Atmospheric
Sciences}, third edition.  Elsevier Academic Press, Amsterdam.

Zhou, X.-H., Obuchowski, N.~A.~and McClish, D.~K.~(2011).  {\em
  Statistical Methods in Diagnostic Medicine}, second edition.  Wiley,
Hoboken.

Zou, K.~H.~and Hall, W.~J.~(2000).  Two transformation models for
estimating an ROC curve derived from continuous data.  {\em Journal of
  Applied Statistics}, {\bf 27}, 621--631.

Zou, K.~H., Resnic, F.~S., Talos, I.-F., Goldberg-Zimring, D.,
Bhagwat, J.~G., Haker, S.~J., Kikinis, R., Jolesz, F.~A.~and
Ohno-Machado, L.~(2005).  A global goodness-of-fit test for receiver
operating characteristic curve analysis via the bootstrap method.
{\em Journal of Biomedical Informatics}, {\bf 38}, 395--403.

Zweig, M.~H.~and Campbell, G.~(1993). Receiver-operating
characteristic (ROC) plots: A fundamental evaluation tool in clinical
medicine. {\em Clinical Chemistry}, {\bf 39}, 561--577.

\end{reflist}

\end{document}